\documentclass[a4paper,11pt,reqno]{amsart}

\newcommand{\LNCSVERSION}[1]{}
\newcommand{\MANUSCRIPTVERSION}[1]{#1}

\newcommand{\mythmspace}{}
\newcommand{\mysecspace}{}
\newcommand{\myfigspace}{}

\usepackage[utf8]{inputenc}
\usepackage[T1]{fontenc}
\usepackage[margin=1in]{geometry}

\usepackage[pdftex,dvipsnames]{xcolor}
\usepackage{tcolorbox}
\usepackage{xargs}
\usepackage[colorinlistoftodos,prependcaption,textsize=normalsize]{todonotes}

\usepackage{amsthm}
\usepackage{amsmath}
\usepackage{amssymb}
\usepackage{thmtools}
\usepackage{thm-restate}

\usepackage{graphicx,transparent,units}
\usepackage{enumitem}
\usepackage[pdftex]{hyperref}

\hypersetup{
pdfauthor={Bartłomiej Bosek, Dariusz Leniowski, Piotr Sankowski, Anna Zych-Pawlewicz},
pdftitle={A Tight Bound for Shortest Augmenting Paths on Trees},
pdfsubject={Computer Science},
pdfkeywords={algorithm, online, incremental, matching, tree}, pdfproducer={LaTeX},
pdfcreator={pdfLaTeX}
}

\newcommand{\executeiffilenewer}[3]{ \ifnum\pdfstrcmp{\pdffilemoddate{#1}}
{\pdffilemoddate{#2}}>0 {\immediate\write18{#3}}\fi }

\declaretheorem[name=Theorem,style=plain]{mytheorem}
\declaretheorem[name=Lemma,sibling=mytheorem,style=plain]{mylemma}
\declaretheorem[name=Claim,sibling=mytheorem,style=plain]{myclaim}
\declaretheorem[name=Corollary,sibling=mytheorem,style=definition]{mycorollary}
\declaretheorem[name=Observation,sibling=mytheorem,style=plain]{myobservation}
\declaretheorem[name=Definition,style=definition]{definition}

\newenvironment{restatablelemma}[3]{}{}

\newenvironment{restatableobservation}[3]{}{}

\LNCSVERSION{\newcommand{\refAppMinimax}{A}}
\LNCSVERSION{\newcommand{\refAppDead}{B}}
\LNCSVERSION{\newcommand{\refAppProof}{C}}

\MANUSCRIPTVERSION{\newcommand{\refAppMinimax}{\ref{app:minimax}}}
\MANUSCRIPTVERSION{\newcommand{\refAppDead}{\ref{app:dead}}}
\MANUSCRIPTVERSION{\newcommand{\refAppProof}{\ref{app:proof}}}

\renewcommand{\leq}{\leqslant}

\renewcommand{\geq}{\geqslant}

\newcommand{\Cases}[4]{\ensuremath{\left\lbrace
\begin{array}
{ll} {#1}&\text{#2}\\
{#3}&\text{#4}
\end{array}
\right.} }
\newcommand{\floor}[1]{\ensuremath{\left \lfloor {#1} \right \rfloor}}

\newcommand{\bra}[1]{\ensuremath{({#1})}}

\newcommand{\sqbra}[1]{\ensuremath{[{#1}]}}
\newcommand{\abs}[1]{\ensuremath{| {#1} |}}

\newcommand{\length}[1]{\ensuremath{{| \! |} {#1} {| \! |}}}

\newcommand{\size}[1]{\ensuremath{| {#1} |}}

\newcommand{\set}[1]{\ensuremath{\{{#1}\}}}
\newcommand{\dedge}[2]{\ensuremath{({#1},{#2})}}

\newcommand{\tuple}[1]{\ensuremath{\left \langle {#1} \right \rangle}}
\newcommand{\range}[1]{\ensuremath{[#1]}}
\newcommand{\neighbours}[2]{\ensuremath{\operatorname{N}_{#1} \bra{#2}}}

\newcommand{\infinity}{\ensuremath{\infty}}

\newcommand{\treeDist}[2]{\ensuremath{\operatorname{mini-max}_{#1} \bra{{#2}}}}
\newcommand{\treeDir}[2]{\ensuremath{\operatorname{mini-max-next}_{#1} \bra{{#2}}}}

\newcommand{\detDist}[3]{\ensuremath{\operatorname{det-dist}_{#1} \bra{{#2},{#3}}}}

\newcommand{\dist}[2]{\ensuremath{\operatorname{dist}_{#1} \bra{{#2}}}}

\newcommand{\secondDist}[2]{\ensuremath{\operatorname{sec-dist}_{#1} \bra{#2}}}
\newcommand{\detDir}[3]{\ensuremath{\operatorname{det-dir}_{#1} \bra{{#2},{#3}}}}

\newcommand{\dir}[2]{\ensuremath{\operatorname{dir}_{#1} \bra{{#2}}}}

\newcommand{\secondDir}[2]{\ensuremath{\operatorname{sec-dir}_{#1} \bra{#2}}}
\newcommand{\level}[2]{\ensuremath{\operatorname{level}_{#1} \bra{#2}}}
\newcommand{\lev}{\ensuremath{l}}
\newcommand{\levlow}{\ensuremath{\underline{l}}}
\newcommand{\levhi}{\ensuremath{\overline{l}}}
\newcommand{\detPath}[3]{\ensuremath{\operatorname{det-path}_{#1} \bra{{#2},{#3}}}}

\newcommand{\pathh}[2]{\ensuremath{\operatorname{path}_{#1} \bra{{#2}}}}
\newcommand{\pathhp}[2]{\ensuremath{\operatorname{path}^{p}_{#1} \bra{{#2}}}}
\newcommand{\pathhs}[2]{\ensuremath{\operatorname{path}^s_{#1} \bra{{#2}}}}

\newcommand{\mmpath}[2]{\ensuremath{\operatorname{mini-max-path}_{#1}\bra{{#2}}}}
\newcommand{\secondPath}[2]{\ensuremath{\operatorname{sec-path}_{#1} \bra{{#2}}}}

\newcommand{\CComp}{\mathcal{C}}
\newcommand{\component}{\ensuremath{C}}
\newcommand{\comp}[2]{\ensuremath{\operatorname{comp} \bra{{#1},{#2}}}}
\newcommand{\Comp}[1]{\ensuremath{\operatorname{\CComp} \bra{#1}}}

\newcommand{\new}[1]{\ensuremath{\operatorname{\Delta}}}

\newcommand{\TreeTwo}{\ensuremath{S}}
\newcommand{\RootedSubTree}[2]{\ensuremath{\Delta}_{#1}(#2)}
\newcommand{\RootedTreeOne}{\ensuremath{T}}
\newcommand{\RootedTreeTwo}{\ensuremath{S}}
\newcommand{\RootedTreeThree}{\ensuremath{R}}
\newcommand{\RootedTreeROne}{\ensuremath{T_{\curvearrowright}}}

\newcommand{\Forest}{\ensuremath{F}}
\newcommand{\RootedForest}{\ensuremath{F}}

\newcommand{\TreeLevell}[2]{\ensuremath{F_{#1}^{\,#2}}}
\newcommand{\Edges}{\ensuremath{E}}
\newcommand{\EdgesIncidentTo}[1]{\ensuremath{E\bra{#1}}}
\newcommand{\lp}[1]{\ensuremath{\mathcal{LP}_{#1}}}

\newcommand{\Vertices}{\ensuremath{V}}
\newcommand{\VerticesOf}[1]{\ensuremath{V\!\bra{#1}}}

\newcommand{\parent}[2]{\ensuremath{\operatorname{parent}_{#1}\bra{#2}}}
\newcommand{\children}[2]{\ensuremath{\operatorname{Ch}_{#1}\bra{#2}}}

\newcommand{\InducedTree}[1]{\ensuremath{F \sqbra{#1}}}
\newcommand{\InducedEdges}[1]{\ensuremath{E \bra{#1}}}

\newcommand{\pathOne}{\ensuremath{\pi}}
\newcommand{\pathTwo}{\ensuremath{\rho}}

\newcommand{\White}{\ensuremath{W}}
\newcommand{\Black}{\ensuremath{B}}

\newcommand{\Alive}{\ensuremath{A}}
\newcommand{\Dead}{\ensuremath{D}}

\newcommand{\edge}{\ensuremath{e}}

\newcommand{\white}{\ensuremath{w}}
\newcommand{\black}{\ensuremath{b}}

\newcommand{\ttime}{\ensuremath{t}}
\newcommand{\firstVertex}{\ensuremath{v}}
\newcommand{\secondVertex}{\ensuremath{u}}
\newcommand{\thirdVertex}{\ensuremath{x}}

\newcommand{\constDist}{\ensuremath{\beta}}

\newcommand{\constTokens}{\ensuremath{\delta}}
\newcommand{\constThreshold}{\ensuremath{\rho}}

\newcommand{\disp}[1]{\ensuremath{\check{b}_{#1}}}
\newcommand{\sap}{Shortest Augmenting Path algorithm}

\newcommand{\bigo}[1]{\ensuremath{\mathcal{O}\bra{#1}}}

\newcommand{\WA}[2]{\ensuremath{\operatorname{WA}_{#1}\bra{#2}}}

\newcommand{\hc}{\text{Hall’s condition}}

\makeatletter
\newcommand{\definetitlefootnote}[1]{
\newcommand{\addtitlefootnote}{$^{*}$\footnote{\protect\@titlefootnotetext}}
\newcommand\@titlefootnotetext{\spaceskip=\z@skip $^{*}$#1} } \makeatother

\makeatletter
\begin{document}
\definetitlefootnote{\textrm{The work of all authors was supported by Polish National Science Center grant 2013/11/D/ST6/03100. Additionally, the work of P.~Sankowski was partially supported by the project TOTAL (No 677651) that has received funding from ERC.}}
\title[A Tight Bound for Shortest Augmenting Paths on Trees]{A Tight Bound for Shortest Augmenting Paths on Trees\addtitlefootnote}
\author[B.{ }Bosek]{Bart\l{}omiej Bosek}
\author[D.{ }Leniowski]{Dariusz Leniowski}
\author[P.{ }Sankowski]{Piotr Sankowski}
\author[A.{ }Zych-Pawlewicz]{Anna Zych-Pawlewicz}

\address{B.{ }Bosek: Theoretical Computer Science Department, Faculty of Mathematics and Computer Science, Jagiellonian University, Krak\'{o}w, Poland}
\email{bosek@tcs.uj.edu.pl}
\address{D.{ }Leniowski, P.{ }Sankowski, and A.{ }Zych-Pawlewicz: Institute of Computer Science, University of Warsaw, Poland}
\email{d.leniowski@mimuw.edu.pl}
\email{sank@mimuw.edu.pl}
\email{anka@mimuw.edu.pl}

\begin{abstract}
The shortest augmenting path technique is one of the fundamental ideas used in maximum matching and maximum flow algorithms. Since being introduced by  Edmonds and Karp in 1972, it has been widely applied in many different settings. Surprisingly, despite this extensive usage, it is still not well understood even in the simplest case: online bipartite matching problem on trees. In this problem a bipartite tree $T=(\White \uplus\Black, E)$ is being revealed online, i.e., in each round one vertex from $\Black$ with its incident edges arrives. It was conjectured by Chaudhuri et. al. \cite{DBLP:conf/infocom/ChaudhuriDKL09} that the total length of all shortest augmenting paths found is $O(n \log n)$. In this
paper we prove a tight $O(n \log n)$ upper bound for the total length of shortest augmenting paths for trees improving over $O(n \log^2 n)$ bound \cite{DBLP:conf/waoa/BosekLSZ15}.
\end{abstract}
 \maketitle
\mysecspace{}

\section{Introduction}

\mysecspace{}

One of the most fundamental techniques used to solve maximum matchings or flow problems is the augmenting path technique. It augments the solution
along residual paths until the maximum size matching/flow is found. Intuitively, the work needed for that should be minimized if shortest paths are chosen each time. In particular, 
this was the key concept that allowed Edmonds and Karp in 1972 to show the first strongly polynomial time
algorithm for the maximum flow problem~\cite{DBLP:journals/jacm/EdmondsK72}.  Since then it has been widely applied. Surprisingly, despite this effort, it is still not well understood even
in the simplest case --- online bipartite matching problem on trees. This may be due to the fact that shortest augmenting paths do not seem to have strong enough structure admitting exact analysis. Other methods for choosing augmenting paths are easier to analyze \cite{DBLP:conf/focs/BosekLSZ14,DBLP:conf/infocom/ChaudhuriDKL09}. Our work is meant as a step forward towards understanding the shortest augmenting path method for computing the matching on bipartite graphs.

To be able to analyze this approach we adopt the following model. Let $\White$ and $\Black$ be the bipartition of vertices over which the tree will be formed. The set $\White$ (called white vertices) is given up front to the algorithm, whereas the vertices
in $\Black$ (black vertices) arrive online.
We denote by $\Forest_{\ttime} = \tuple{ \White \uplus \Black_{\ttime}, \Edges_{\ttime} }$
the~forest after the $t$'th black vertex has arrived where $X \uplus Y$ is a disjoint sum
of $X$ and $Y$.
The graphs $\Forest_{\ttime}$ for $t \in \range{n}=\set{1,\ldots,n}$ are constructed online in the following manner.
We start with $\Forest_0 = \tuple{ \White \uplus \Black_0, \Edges_0 } = \tuple{ \White
\uplus \emptyset, \emptyset }$.
In turn $t\in \range{n} $ a new vertex $\black_{\ttime} \in \Black$ together with all
its incident edges $\EdgesIncidentTo{\black_{\ttime}}$ is revealed and $\Forest_{\ttime}$ is
defined as: $\Edges_{\ttime} = \Edges_{t-1}\cup \EdgesIncidentTo{\black_{\ttime}}$ and
$\Black_{\ttime} = \Black_{t-1}\cup \set{\black_{\ttime}}$.
In the model we consider, none of the newly added edges is allowed to close the cycle.
For simplicity we assume that we add in total $n=|\White|$ black vertices.
The final graph is a tree denoted as $\Forest_{n}=(\White\uplus \Black_n,\Edges_n)$.

The goal of the online algorithm is to compute for each $\Forest_{\ttime}$ the maximum size
matching $M_{\ttime}$, possibly making use of $M_{t-1}$.
In this paper we study one specific algorithm, referred to as the Shortest Augmenting Path
algorithm.
When $\black_\ttime$ arrives, the \sap{} always chooses the shortest among all
available augmenting paths. A natural question that we ask is what is the total length of all paths applied by the \sap{}.
In this paper the unmatched vertices are referred to as \emph{free}. For a vertex $v$ we denote its neighborhood in $\Forest_{\ttime}$ as $\neighbours{\ttime}{v}$. By $\InducedTree{X}= \tuple{X, \InducedEdges{X}}$ we denote a subgraph of $\Forest$ induced by $X
\subseteq \White \cup \Black$, where $\InducedEdges{X}=\set{\edge \in \Edges: \edge
\subseteq X}$.

\mysecspace{}

\section{Motivation and Related Work}

\mysecspace{}

The online bipartite matching problem with augmentations has recently received increasing attention~\cite{DBLP:journals/corr/BernsteinHR17,DBLP:conf/focs/BosekLSZ14,DBLP:conf/waoa/BosekLSZ15,DBLP:conf/infocom/ChaudhuriDKL09,DBLP:conf/wads/GroveKKV95,DBLP:conf/soda/GuptaKS14}. The model we study has been introduced in~\cite{DBLP:conf/wads/GroveKKV95}. As mentioned before, the key point of this model is to focus on bounding the total length of augmenting paths and not the running time of the algorithm. This is motivated as follows. Imagine that the white vertices are servers and black vertices are clients. The clients arrive online. A typical client may be a portable computing box, trying to connect to a huge network of services with some specific request. The edges of the graph reflect eligibility of the servers to answer clients request. The classical online model (as in \cite{DBLP:journals/sigact/BirnbaumM08,DBLP:conf/soda/DevanurJK13,DBLP:conf/stoc/KarpVV90}) does not allow preemption, i.e., the client cannot change the server. In such setting one must accept some clients not being served while they could possibly be served with preemption. In that model a famous ranking algorithm gives an optimal $(1-1/e)$-approximation \cite{DBLP:conf/stoc/KarpVV90}. The authors in~\cite{DBLP:conf/wads/GroveKKV95} wonder if preemption makes sense. It may be beneficial to reallocate clients provided that only a limited number of reallocations is needed. This leads to the question of how many reallocations are needed if one insists on serving every client. In~\cite{DBLP:conf/wads/GroveKKV95} a special case is studied when each client can connect to at most two servers. In such scenario the authors prove that the \sap{} performs $\bigo{n \log n}$ reallocations and that no algorithm can do better than that. Chaudhuri et al. \cite{DBLP:conf/infocom/ChaudhuriDKL09} show that the \sap{} makes a total of $\bigo{n \log n}$ reallocations in the case of general bipartite graph, provided that the clients arrive in a random order. They conjecture, however, that this should be the case also for the worst case arriving order of clients. Until this paper, this conjecture remained open even for trees.
In \cite{DBLP:conf/waoa/BosekLSZ15} the authors prove a bound of $\bigo{n \log^2 n}$ for \sap{} given that the underlying graph is a tree. In this paper we take a different approach and prove the conjecture of Chaudhuri et al. for trees. In this restricted case, the authors of~\cite{DBLP:conf/infocom/ChaudhuriDKL09} proposed a different augmenting path algorithm that achieves total paths' length of $\bigo{n \log n}$. Their algorithm, however, is only applicable to trees. The \sap{}, on the other hand, applies to any bipartite graph and also is very simple. This is the reason why we feel it is important to study this algorithm.
Our ultimate goal is to show the bound of $\bigo{n \log n}$ for general bipartite graphs. We believe that the techniques proposed in this paper are an important step forward on the path to achieve this goal.
In parallel work to ours \cite{DBLP:journals/corr/BernsteinHR17} the authors provide a bound of $\bigo{n \log^2 n}$ total number of reallocations for the \sap{} on general bipartite graphs. This recent result has been accepted to \emph{SODA 2018} and it nearly closes the conjecture of Chaudhuri et al. We note, however, that their techniques alone do not lead to $\bigo{n \log n}$ even for trees. 
Before this result, for general graphs, nothing interesting was known for \sap{}. A different algorithm was proposed achieving much worse $\bigo{n \sqrt{n}}$ bound on the total length of augmenting paths~\cite{DBLP:conf/focs/BosekLSZ14}.

Our model is strongly related to dynamic algorithms.
There, we are not only interested in constructing short augmenting paths. An efficient way of finding 
them is the most important aspect.
Most papers in this area consider edge updates in a general fully-dynamic model
which allows for insertions and deletions of edges intermixed with each other. This is a much more difficult scenario in which one cannot do much when constrained by our model. In particular, if edges are added to a bipartite graph, one can show an instance for which any algorithm maintaining a maximum matching performs $\Omega(n^2)$ reallocations. Hence, it is reasonable to stop insisting on matching every client and accept approximate solutions. Here we want to approximate the maximum matching size and not the number of reallocations.
One also needs to keep in mind that a trivial greedy algorithm maintaining a maximal matching gives a $1/2$-approximation and
preforms no reallocations at all. 
A $2/3$-approximation algorithm by~\cite{DBLP:conf/stoc/NeimanS13} achieves $\bigo{\sqrt{m}}$
update time. Gupta and Peng give a $(1-\varepsilon)$-approximation in $\bigo{\sqrt{m}\varepsilon^{-2}}$ time per update~\cite{DBLP:conf/focs/GuptaP13}.
The $\bigo{\sqrt{m}}$ barrier was broken by Bernstein and Stein who gave a $(\frac{2}{3}-\varepsilon)$-approximation
algorithm that achieves $O(m^{1/4}\varepsilon^{-2.5})$ update time~\cite{DBLP:conf/icalp/BernsteinS15}.
Finally,  $(1-\varepsilon)$ approximation
in $O(m\varepsilon^{-1})$ total time and with $O(n\varepsilon^{-1})$ total length of paths was shown in~\cite{DBLP:conf/focs/BosekLSZ14} in a model most related to ours, i.e., when vertices are added on one side of the bipartition. There are also randomized algorithms in the dynamic model~\cite{DBLP:conf/soda/Sankowski07} maintaining the exact size of a maximum matching with $\bigo{n^{1.495}}$ update time. They
do not imply any bound on the number of changes to the matching as they use algebraic techniques that are
not based on augmenting paths.
 \mysecspace{}

\section{The mini-max game}\label{sec:minimax}

\mysecspace{}

Our goal in this paper is to prove that the total length of all augmenting paths applied by \sap{} on a tree is $O(n \log n)$.
More formally, we want to prove the following, where by $\length{\pathOne}$ we denote the number
of edges on a path $\pathOne$.

\mythmspace{}

\begin{mytheorem}
\label{thm:mainthm}
Let $\pathOne_\ttime$ be the path applied by \sap{} in turn $\ttime$.
Then
$
\sum_{\ttime=1}^{n} \length{\pathOne_\ttime} \in O(n \log n).
$
\end{mytheorem}

\mythmspace{}

The idea is not to study directly the paths applied by \sap{}, but a collection of other paths that are possibly longer. To be more precise, we model a scenario where in each turn \sap{} gets the worst
possible matching on $\Forest_{\ttime}$ (i.e., the one maximizing the shortest augmenting path).
We then study the worst case augmenting paths rather then the ones given by the matching produced by \sap{}. Interestingly, these paths can be defined without mentioning any matching. In this section we provide the appropriate definitions and show that they work as expected.

Let us consider what worst possible matching could there be. 
Think of a game, where the algorithm chooses a shortest augmenting path, and the adversary chooses a matching where such path is the longest. We are given graph $\Forest_{\ttime}$ and the newly presented vertex $\black_{\ttime}$. We are interested in a matching where $\black_{\ttime}$ is not matched, so that we model the worst case matching before $\black_{\ttime}$ is matched. The game starts in vertex $\black_{\ttime}$, where the algorithm may choose which edge to follow among the unmatched edges incident to $\black_{\ttime}$. Then the algorithm stumbles upon a white vertex where it has to follow the matching edge chosen by the adversary. The game continues until a leaf is reached (either black or white, black meaning that the algorithm did not find a path). It is not hard to see that the algorithm, when it has a choice, wants to minimize the distance to a free white vertex, while the adversary tries to maximize it. This way we obtain a two-person game, where the outcome of the game is the length of the shortest augmenting path. If the path does not exist, we let the outcome be infinite. Throughout the paper we let $\infinity+1=\infinity$ and we write $x < \infinity$ to indicate that $x$ is simply an integer.

We move on to stating formal definitions. We start by introducing our game on any rooted tree $\RootedTreeOne$ whose vertices are either black or white. We then define the outcome of the algorithm player in time $\ttime$ for a specific $\RootedTreeOne$ closely related to $\Forest_{\ttime}$. For a rooted tree $\RootedTreeOne$ we denote the list of children of a vertex $v$ in $\RootedTreeOne$ as $\children{\RootedTreeOne}{v}$ and a parent of $v$ as $\parent{\RootedTreeOne}{v}$.

\mythmspace{}

\begin{definition}\label{def:minimax}
\MANUSCRIPTVERSION{
Let $\RootedTreeOne$ be a rooted tree whose vertices are partitioned into two sets:
$V(\RootedTreeOne)\subseteq\Black \uplus \White$.
\begin{enumerate}
\item For each $\black\in\Black$ we define its revenue as
$$
\treeDist{\RootedTreeOne}{\black} = \left\lbrace
\begin{array}
{ll} \min_{\white \in \children{\RootedTreeOne}{\black}}
\treeDist{\RootedTreeOne}{\white}+1 & \text{if} \ \children{\RootedTreeOne}{\black}\neq \emptyset \\
\infinity & \text{otherwise}
\end{array}
\right.
$$
\item and for each $\white\in\White$ we define its revenue as
$$
\treeDist{\RootedTreeOne}{\white} = \left\lbrace
\begin{array}
{ll} \max_{\black \in \children{\RootedTreeOne}{\white}}
\treeDist{\RootedTreeOne}{\black}+1 & \text{if} \ \children{\RootedTreeOne}{\white}\neq \emptyset \\
0 & \text{otherwise.}
\end{array}
\right.
$$
\item We let $\treeDir{\RootedTreeOne}{\firstVertex}$ be the child of $\firstVertex$ whose revenue determines the minimum or the maximum respectively.\!\footnotemark[1]\footnotetext[1]{If there are more such vertices we choose the first one according to some predefined order on $\Black\cup\White$.} If $\firstVertex$ has no children, $\treeDir{\RootedTreeOne}{\firstVertex}$ is undefined.

\item We define the mini-max path starting in a vertex $\firstVertex$ as\footnotetext[2]{Symbol $\cdot$ denotes concatenation of paths.}
\begin{equation*}
\mmpath{\RootedTreeOne}{\firstVertex}=
\left \lbrace
\begin{array}{ll}
\firstVertex\cdot \mmpath{\RootedTreeOne}{\treeDir{\RootedTreeOne}{\firstVertex}} & \text{if} \ \treeDir{\RootedTreeOne}{\firstVertex} \ \text{is defined,\!\footnotemark[2]} \\
\firstVertex & \text{otherwise.}
\end{array}
\right.
\end{equation*}
\end{enumerate}
}
\LNCSVERSION{
Let $\RootedTreeOne$ be a rooted tree whose vertices are partitioned into two sets:
$V(\RootedTreeOne)\subseteq\Black \uplus \White$.
For each $\black\in\Black$ we define its revenue as $\treeDist{\RootedTreeOne}{\black}
= \min_{\white \in \children{\RootedTreeOne}{\black}} \treeDist{\RootedTreeOne}{\white}+1$
if $\children{\RootedTreeOne}{\black}\neq \emptyset$ and $\treeDist{\RootedTreeOne}{\black} = \infinity$ otherwise.
For each $\white\in\White$ we define its revenue as
$\treeDist{\RootedTreeOne}{\white} = \max_{\black \in \children{\RootedTreeOne}{\white}}
\treeDist{\RootedTreeOne}{\black}+1$ if $\children{\RootedTreeOne}{\white}\neq \emptyset$ and $\treeDist{\RootedTreeOne}{\white}=0$ otherwise.
We let $\treeDir{\RootedTreeOne}{\firstVertex}$ be the child of $\firstVertex$ whose revenue determines the minimum or the maximum respectively.\!\footnotemark[1]\footnotetext[1]{If there are more such vertices we choose the first one according to some predefined order on $\Black\cup\White$.} If $\firstVertex$ has no children, $\treeDir{\RootedTreeOne}{\firstVertex}$ is undefined.
We define the mini-max path starting in a vertex $\firstVertex$ as
$\mmpath{\RootedTreeOne}{\firstVertex}=
\firstVertex \cdot \mmpath{\RootedTreeOne}{\treeDir{\RootedTreeOne}{\firstVertex}}$ if $\treeDir{\RootedTreeOne}{\firstVertex}$ is defined and $\mmpath{\RootedTreeOne}{\firstVertex}=
\firstVertex$ otherwise.\!\footnotemark[2]\footnotetext[2]{Symbol $\cdot$ denotes concatenation of paths.}
}
\end{definition}

\mythmspace{}

\MANUSCRIPTVERSION{
\begin{figure}[htbp]\fontsize{10}{10}\selectfont
\centering \def\svgscale{0.9} \resizebox{0.7\textwidth}{!}{
\executeiffilenewer{rooted_tree.svg}{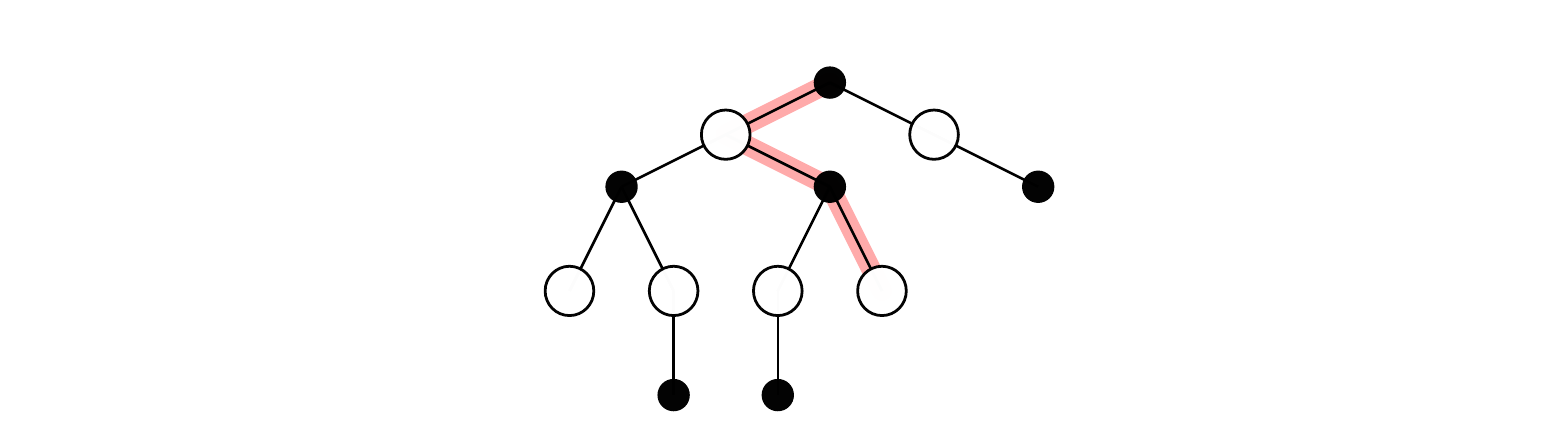} {inkscape -z -C --file=rooted_tree.svg --export-pdf=rooted_tree.pdf
--export-latex}
\begingroup
  \makeatletter
  \providecommand\color[2][]{
    \errmessage{(Inkscape) Color is used for the text in Inkscape, but the package 'color.sty' is not loaded}
    \renewcommand\color[2][]{}
  }
  \providecommand\transparent[1]{
    \errmessage{(Inkscape) Transparency is used (non-zero) for the text in Inkscape, but the package 'transparent.sty' is not loaded}
    \renewcommand\transparent[1]{}
  }
  \providecommand\rotatebox[2]{#2}
  \ifx\svgwidth\undefined
    \setlength{\unitlength}{446.34613758bp}
    \ifx\svgscale\undefined
      \relax
    \else
      \setlength{\unitlength}{\unitlength * \real{\svgscale}}
    \fi
  \else
    \setlength{\unitlength}{\svgwidth}
  \fi
  \global\let\svgwidth\undefined
  \global\let\svgscale\undefined
  \makeatother
  \begin{picture}(1,0.28441532)
    \put(0,0){\includegraphics[width=\unitlength,page=1]{rooted_tree.pdf}}
    \put(0.25307288,0.12336361){\color[rgb]{0,0,0}\makebox(0,0)[lb]{\smash{$\RootedTreeOne$:}}}
    \put(0.56309072,0.0889172){\color[rgb]{0,0,0}\makebox(0,0)[lb]{\smash{$0$}}}
    \put(0.3614533,0.0889172){\color[rgb]{0,0,0}\makebox(0,0)[lb]{\smash{$0$}}}
    \put(0.42236461,0.09059754){\color[rgb]{0,0,0}\makebox(0,0)[lb]{\smash{$\infty$}}}
    \put(0.48957709,0.09059754){\color[rgb]{0,0,0}\makebox(0,0)[lb]{\smash{$\infty$}}}
    \put(0.5903958,0.19141625){\color[rgb]{0,0,0}\makebox(0,0)[lb]{\smash{$\infty$}}}
    \put(0.46227201,0.18973601){\color[rgb]{0,0,0}\makebox(0,0)[lb]{\smash{$2$}}}
    \put(0.50932094,0.04270854){\color[rgb]{0,0,0}\makebox(0,0)[lb]{\smash{$\infty$}}}
    \put(0.44210846,0.04270854){\color[rgb]{0,0,0}\makebox(0,0)[lb]{\smash{$\infty$}}}
    \put(0.39505954,0.17713349){\color[rgb]{0,0,0}\makebox(0,0)[lb]{\smash{$1$}}}
    \put(0.52948449,0.17713349){\color[rgb]{0,0,0}\makebox(0,0)[lb]{\smash{$1$}}}
    \put(0.52948449,0.24434597){\color[rgb]{0,0,0}\makebox(0,0)[lb]{\smash{$3$}}}
    \put(0.66054877,0.17713349){\color[rgb]{0,0,0}\makebox(0,0)[lb]{\smash{$\infty$}}}
  \end{picture}
\endgroup
  }
\caption{Example of a rooted mini-max tree with vertex revenues.}
\label{fig:rooted_tree}
\end{figure}
}

Definition~\ref{def:minimax} is illustrated by example in Figure~\MANUSCRIPTVERSION{\ref{fig:rooted_tree}}\LNCSVERSION{\ref{fig:rooted_tree_and_dir_sec_dir}.(a)}.
Based on this definition we can define
the first and the second 
mini-max distance from a given vertex to a white leaf in a specific time moment $\ttime$. In addition to that we define the first and second direction, i.e., the vertex one needs to follow to find the first and second mini-max distance.

\mythmspace{}

\begin{definition}\label{def:distdir}
\MANUSCRIPTVERSION{
Let $\ttime \in \range{n} $ and $v \in \Black_{\ttime} \cup \White$. Let $\RootedTreeOne$ be the connected component of $v$ in $\Forest_{\ttime}$ rooted in $v$. Let 
\begin{align*}
\dist{\ttime}{v} & =\treeDist{\RootedTreeOne}{v},\\
\dir{\ttime}{v} & =\treeDir{\RootedTreeOne}{v}.
\end{align*}
Let now $\RootedTreeTwo$ be a rooted tree, where from $\RootedTreeOne$ we remove $\treeDir{\RootedTreeOne}{\firstVertex}$ and all its descendants. Let
\begin{align*}
\secondDist{\ttime}{v} & =\treeDist{\RootedTreeTwo}{v},\\
\secondDir{\ttime}{v} & =\treeDir{\RootedTreeTwo}{v}.
\end{align*}
}
\LNCSVERSION{
Let $\ttime \in \range{n} $ and $v \in \Black_{\ttime} \cup \White$. Let $\RootedTreeOne$ be the connected component of $v$ in $\Forest_{\ttime}$ rooted in $v$. Let 
$\dist{\ttime}{v}  =\treeDist{\RootedTreeOne}{v}$ and
$\dir{\ttime}{v}  =\treeDir{\RootedTreeOne}{v}$.
Let now $\RootedTreeTwo$ be a rooted tree, where from $\RootedTreeOne$ we remove $\treeDir{\RootedTreeOne}{\firstVertex}$ and all its descendants. Let
$\secondDist{\ttime}{v} =\treeDist{\RootedTreeTwo}{v}$ and
$\secondDir{\ttime}{v} =\treeDir{\RootedTreeTwo}{v}$.
}
\end{definition}

\mythmspace{}

\MANUSCRIPTVERSION{
\begin{figure}[htbp]\fontsize{10}{10}\selectfont
\centering \def\svgscale{0.9} \resizebox{0.65\textwidth}{!}{
\executeiffilenewer{dir_sec_dir.svg}{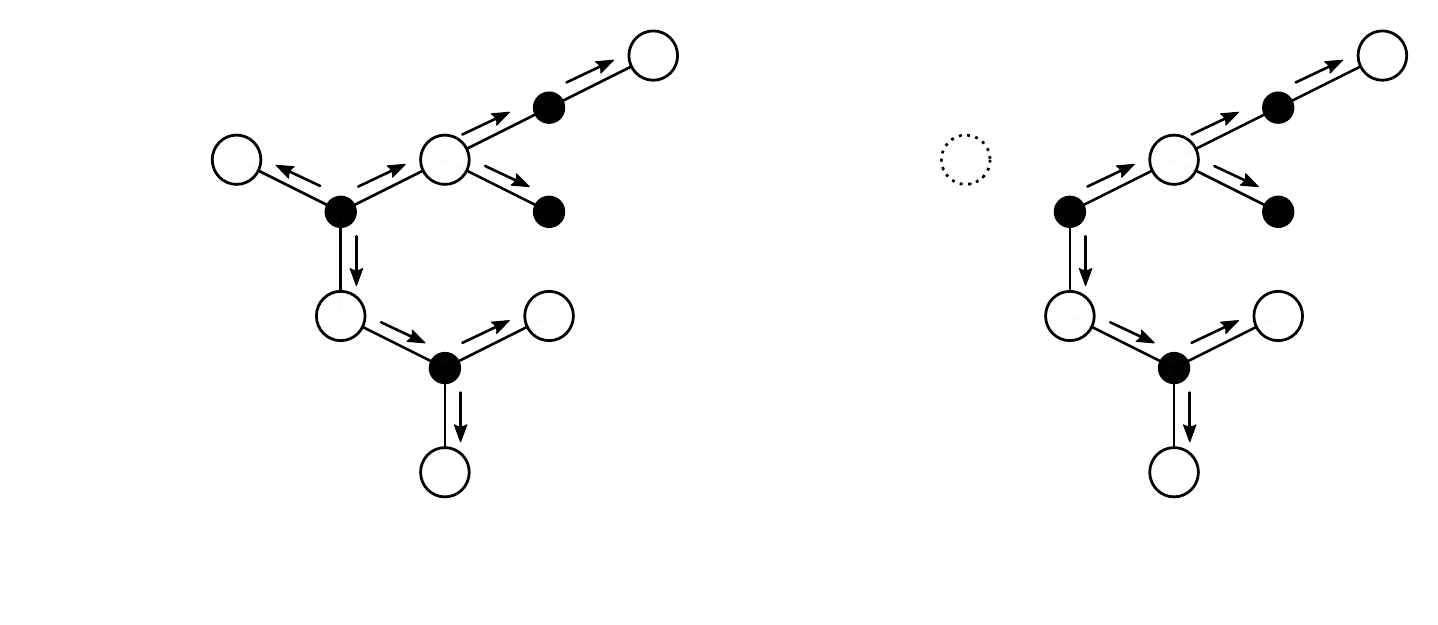} {inkscape -z -C --file=dir_sec_dir.svg --export-pdf=dir_sec_dir.pdf
--export-latex}
\begingroup
  \makeatletter
  \providecommand\color[2][]{
    \errmessage{(Inkscape) Color is used for the text in Inkscape, but the package 'color.sty' is not loaded}
    \renewcommand\color[2][]{}
  }
  \providecommand\transparent[1]{
    \errmessage{(Inkscape) Transparency is used (non-zero) for the text in Inkscape, but the package 'transparent.sty' is not loaded}
    \renewcommand\transparent[1]{}
  }
  \providecommand\rotatebox[2]{#2}
  \ifx\svgwidth\undefined
    \setlength{\unitlength}{414.04637615bp}
    \ifx\svgscale\undefined
      \relax
    \else
      \setlength{\unitlength}{\unitlength * \real{\svgscale}}
    \fi
  \else
    \setlength{\unitlength}{\svgwidth}
  \fi
  \global\let\svgwidth\undefined
  \global\let\svgscale\undefined
  \makeatother
  \begin{picture}(1,0.43471636)
    \put(0,0){\includegraphics[width=\unitlength,page=1]{dir_sec_dir.pdf}}
    \put(0.2948562,0.02629828){\color[rgb]{0,0,0}\makebox(0,0)[lb]{\smash{$\Forest_{\ttime}$}}}
    \put(0.75132735,0.02629828){\color[rgb]{0,0,0}\makebox(0,0)[lb]{\smash{$\Forest_{\ttime} \setminus \{v, \dir{t}{v}\}$}}}
    \put(0.02943232,0.3039718){\color[rgb]{0,0,0}\makebox(0,0)[lb]{\smash{$\dir{t}{v}$}}}
    \put(0.05575201,0.24004255){\color[rgb]{0,0,0}\makebox(0,0)[lb]{\smash{$v$}}}
    \put(0.60641591,0.26177927){\color[rgb]{0,0,0}\makebox(0,0)[lb]{\smash{$v$}}}
    \put(0.55071592,0.18117237){\color[rgb]{0,0,0}\makebox(0,0)[lb]{\smash{$\secondDir{t}{v}$}}}
    \put(0.15809635,0.3152153){\color[rgb]{0,0,0}\makebox(0,0)[lb]{\smash{$0$}}}
    \put(0.23055207,0.30162981){\color[rgb]{0,0,0}\makebox(0,0)[lb]{\smash{$1$}}}
    \put(0.37184068,0.30162981){\color[rgb]{0,0,0}\makebox(0,0)[lb]{\smash{$\infty$}}}
    \put(0.29621508,0.31702671){\color[rgb]{0,0,0}\makebox(0,0)[lb]{\smash{$\infty$}}}
    \put(0.37546352,0.37408553){\color[rgb]{0,0,0}\makebox(0,0)[lb]{\smash{$1$}}}
    \put(0.44791924,0.38767091){\color[rgb]{0,0,0}\makebox(0,0)[lb]{\smash{$0$}}}
    \put(0.37546352,0.20653182){\color[rgb]{0,0,0}\makebox(0,0)[lb]{\smash{$0$}}}
    \put(0.3030078,0.09784829){\color[rgb]{0,0,0}\makebox(0,0)[lb]{\smash{$0$}}}
    \put(0.3030078,0.19294633){\color[rgb]{0,0,0}\makebox(0,0)[lb]{\smash{$1$}}}
    \put(0.23055207,0.20653182){\color[rgb]{0,0,0}\makebox(0,0)[lb]{\smash{$2$}}}
    \put(0,0){\includegraphics[width=\unitlength,page=2]{dir_sec_dir.pdf}}
    \put(0.73774212,0.30162981){\color[rgb]{0,0,0}\makebox(0,0)[lb]{\smash{$3$}}}
    \put(0.87903031,0.30162981){\color[rgb]{0,0,0}\makebox(0,0)[lb]{\smash{$\infty$}}}
    \put(0.80340492,0.31702671){\color[rgb]{0,0,0}\makebox(0,0)[lb]{\smash{$\infty$}}}
    \put(0.88265315,0.37408553){\color[rgb]{0,0,0}\makebox(0,0)[lb]{\smash{$1$}}}
    \put(0.95510866,0.38767091){\color[rgb]{0,0,0}\makebox(0,0)[lb]{\smash{$0$}}}
    \put(0.88265315,0.20653182){\color[rgb]{0,0,0}\makebox(0,0)[lb]{\smash{$0$}}}
    \put(0.81019764,0.09784829){\color[rgb]{0,0,0}\makebox(0,0)[lb]{\smash{$0$}}}
    \put(0.81019764,0.19294633){\color[rgb]{0,0,0}\makebox(0,0)[lb]{\smash{$1$}}}
    \put(0.73774212,0.20653182){\color[rgb]{0,0,0}\makebox(0,0)[lb]{\smash{$2$}}}
  \end{picture}
\endgroup
  }
\caption{Example of mini-max distances for a vertex $v$ in turn $\ttime$.}
\label{fig:dir_sec_dir}
\end{figure}
}

Definition~\ref{def:distdir} is illustrated by example in Figure~\MANUSCRIPTVERSION{\ref{fig:dir_sec_dir}}\LNCSVERSION{\ref{fig:rooted_tree_and_dir_sec_dir}.(b)}.
We next observe the monotonicity of the mini-max distance functions we
defined.

\mythmspace{}

\begin{restatable}{myobservation}{obsmonotonic}
\begin{restatableobservation}[{{\rm Appendix~\refAppMinimax{} in \cite{DBLP:journals/corr/BosekLZS17}}}]
\label{obs:monotonic}
Fix a vertex $u \in \Black \cup \White$. The functions $\dist{\ttime}{u}$ and $\secondDist{\ttime}{u}$ are non-decreasing with respect to $\ttime$ for the whole range of $\ttime$ where $u \in V(\Forest_{\ttime})$.
\end{restatableobservation}
\end{restatable}

\mythmspace{}

\MANUSCRIPTVERSION{
\begin{proof}
Observe that if $\RootedTreeOne$ is a component of $\firstVertex$ in $\Forest_{\ttime-1}$ rooted in $\firstVertex$, then adding $\black_{\ttime}$ causes that $\black_{\ttime}$ possibly becomes a child of some white vertex $\white \in \Vertices\bra{\RootedTreeOne}$ (see Figure~\ref{fig:ststr12} to the left in Appendix~\refAppDead{}).
Such change causes that either we take maximum over a larger set than before, or we take maximum/minimum over a set whose values do not decrease. In all these cases we cannot decrease the mini-max values. For the rigorous proof see Appendix~\refAppMinimax{}.
\end{proof}
}

We intuitively explained how the mini-max distances correspond to the augmenting paths of the worst case matching, so the hope is that they bound from above the augmenting paths applied by \sap{}. The next lemma shows that this intuition is reflected in reality.
It states that no matter what matching is given on $\Forest_{\ttime}$ for some $\ttime \in \range{n}$, the path chosen by \sap{} to match $\black_{\ttime}$ is bounded by $\dist{\ttime}{\black_{\ttime}}$.

\mythmspace{}

\begin{restatable}{mylemma}{lempathineq}
\begin{restatablelemma}[{{\rm Appendix~\refAppMinimax{} in \cite{DBLP:journals/corr/BosekLZS17}}}]
\label{lem:pathineq}
Let $1\leqslant \ttime_0 \leqslant \ttime \leqslant n$ and let $\pathTwo_t$ be the shortest augmenting path from $\black_{\ttime_0}$ to a free white
vertex according to any given matching $M$ in $\Forest_{\ttime}$ where $\black_{\ttime_0}$ is free. It holds that if $\dist{\ttime}{\black_{\ttime_0}} < \infinity$ then $\pathTwo_t$ exists and
$
\length{\pathTwo_t} \leqslant \dist{\ttime}{\black_{\ttime_0}}.
$
\end{restatablelemma}
\end{restatable}

\mythmspace{}

\MANUSCRIPTVERSION{
\begin{proof}
For the proof we refer to Appendix~\refAppMinimax{}\LNCSVERSION{ in \cite{DBLP:journals/corr/BosekLZS17}}. 
\LNCSVERSION{\qed}
\end{proof}
}

 \mysecspace{}

\section{Dead vertices}\label{sec:dead}

\mysecspace{}

In this section we introduce another concept crucial for our result. We define here \emph{dead} vertices and give some intuition why this makes sense. In fact dead vertices reflect the infinity of some mini-max distance functions. For completeness, in addition to defining dead vertices, we describe the situations when the mini-max distance functions are infinite. We start with the statements.

\mythmspace{}

\begin{definition}\label{def:breakHC}
\MANUSCRIPTVERSION{
A vertex $\black_{\ttime_0} \in \Black_{\ttime}$ breaks \hc{} in time $\ttime \geq \ttime_0$ iff there exists $X\subseteq \Black_{\ttime}$ such that 
\begin{enumerate}
\item $\black_{\ttime_0} \in X$,
\item $\size{\neighbours{\ttime}{X}} < \size{X}$, and
\item $X$ is minimal under inclusion set satisfying (2).
\end{enumerate}
}
\LNCSVERSION{
A vertex $\black_{\ttime_0} \in \Black_{\ttime}$ breaks \hc{} in time $\ttime \geq \ttime_0$ iff $\black_{\ttime_0} \in X$ for some inclusion-wise minimal set $X$ satisfying $\size{\neighbours{\ttime}{X}} < \size{X}$.  
}
\end{definition}

\mythmspace{}

\mythmspace{}

\begin{restatable}{mylemma}{lemdistBlinf}
\begin{restatablelemma}[{{\rm Appendix~\refAppDead{} in \cite{DBLP:journals/corr/BosekLZS17}}}]
\label{lem:distBlinf}
Let $1\leqslant \ttime_0 \leqslant \ttime \leqslant n$. Then
$\dist{\ttime}{\black_{\ttime_0}}=\infinity$ iff $\black_{\ttime_0}$ breaks \hc{} in time $\ttime.$ 
\end{restatablelemma}
\end{restatable}

\mythmspace{}

\MANUSCRIPTVERSION{
\begin{proof}
"$\Leftarrow$" : 
Assume that $\dist{\ttime}{\black_{\ttime_0}}<\infinity$. Then, by Lemma~\ref{lem:pathineq}, an augmenting path exists from $\black_{\ttime_0}$ for a maximum matching $M$ where $\black_{\ttime_0}$ is free. By Hall's theorem this implies that $\black_{\ttime_0}$ does not break \hc{}.
"$\Rightarrow$" : 
Deferred to Appendix~\refAppDead{}\end{proof}

Lemma~\ref{lem:distBlinf} reveals the following corollary. 
}

\mythmspace{}

\begin{mycorollary}\label{cor:nonMatch}
If $\dist{\ttime}{\black_{\ttime}}$ is infinite, and we are given some maximum matching $M_{\ttime-1}$ for $\Forest_{\ttime-1}$, then there is no augmenting (with respect to $M_{\ttime-1}$) path for $\black_{\ttime}$ in time $\ttime$. 
\end{mycorollary}

\mythmspace{}

\MANUSCRIPTVERSION{
\begin{proof}
Follows directly from Hall's theorem. 
\end{proof}
}

We now move on to defining dead vertices. The definitions may not seem very intuitive, but we provide some intuition shortly after introducing them.

\mythmspace{}

\begin{definition}\label{def:deadBlack}
We say that a vertex $\black_{\ttime_0} \in \Black$ is \emph{dead} in turn $\ttime \geq  \ttime_0$ iff $\secondDist{\ttime}{\black_{\ttime_0}} = \infinity.$ 
\end{definition}

\mythmspace{}

Definition~\ref{def:deadBlack} combined with Lemma~\ref{lem:distBlinf} implies, that any black vertex that breaks Hall's condition in time $t$ is dead in time $t$, but not necessarily the other way around.

\mythmspace{}

\begin{definition}\label{def:deadWhite}
A white vertex $\white \in \White$ is \emph{dead} in time $\ttime$  iff  $\dist{\ttime}{\white} = \infinity.$ 
\end{definition}

\mythmspace{}

We say that a vertex is alive iff it is not dead. We denote as $\Alive_{\ttime}$ the set of vertices of $\Black_{\ttime} \cup \White$ that are alive in turn $t$ and as $\Dead_{\ttime}$ the set of vertices of $\Black_{\ttime} \cup \White$ that are dead in turn $\ttime$. If $v \in \Alive_{\ttime-1} \cap \Dead_{\ttime}$, we say that $v$ dies in turn $\ttime$. Note that due to monotonicity (Observation~\ref{obs:monotonic}), once a vertex dies, it never comes back alive.
The following observations bring some intuition into the picture of dead versus alive vertices. Observation~\ref{obs:aliveness} given below follows from Definitions~\ref{def:distdir},~\ref{def:deadBlack}~and~\ref{def:deadWhite}.

\mythmspace{}

\begin{myobservation}\label{obs:aliveness}
\MANUSCRIPTVERSION{$ $}
\begin{enumerate}
 \item A black leaf is dead from the moment it arrives.
 \item A black vertex is dead iff it has at most one alive neighbour.
 \item A white vertex is dead iff it has at least one dead neighbour.
\end{enumerate}
\end{myobservation}

\mythmspace{}

The intuition behind dead vertices is that they determine regions of $\Forest_{\ttime}$ where \hc{} is either broken or tight. The mini-max paths in turn $\ttime$ that correspond to finite mini-max distances do not visit vertices that were dead in turn $\ttime-1$. Moreover, if a mini-max path in $\Forest_\ttime$ (whose corresponding mini-max distance is finite) enters a vertex that is alive in turn $\ttime$, it does not visit anymore vertices dead in turn $\ttime$. We state this formally as Lemma~\ref{lem:mmPathsAlive}. This reflects the behavior of augmenting paths. If a maximum matching is maintained, then the augmenting path from turn $\ttime$ does not enter the regions where \hc{} is tight in $\Forest_{\ttime-1}$. 

\mythmspace{}

\begin{restatable}{mylemma}{lemmmPathsAlive}
\begin{restatablelemma}[{{\rm Appendix~\refAppDead{} in \cite{DBLP:journals/corr/BosekLZS17}}}]
\label{lem:mmPathsAlive}
Let $\ttime\in\range{n} $ and $\firstVertex\in \Alive_{\ttime}$.
Pick any vertex as a root of the connected component of $v$ in $\Forest_{\ttime}$ and let $\RootedTreeOne$ be the corresponding rooted tree. Then $\Vertices\bra{\mmpath{\RootedTreeOne}{\firstVertex}}\subseteq\Alive_{\ttime}$.
\end{restatablelemma}
\end{restatable}

\mythmspace{}

\MANUSCRIPTVERSION{
\begin{proof}
See Appendix~\refAppDead{}.
\end{proof}
}

In the remainder of this section we specify precisely which vertices die in turn $t$. The first lemma does not describe who dies or stays alive, but it is an important complement of the subsequent two lemmas, which cover all the situations when vertices die.

\mythmspace{}

\begin{restatable}{mylemma}{lemaliveneighbour}
\begin{restatablelemma}[{{\rm Appendix~\refAppDead{} in \cite{DBLP:journals/corr/BosekLZS17}}}]
\label{lem:aliveneighbour}
If $\black_{\ttime}$ does not break \hc{} in turn $t$, then $\black_{\ttime}$ has at least one neighbour in $\Forest_{\ttime}$ which was alive in turn $\ttime-1$.
\end{restatablelemma}
\end{restatable}

\mythmspace{}

\MANUSCRIPTVERSION{
\begin{proof}
See Appendix~\refAppDead{}.
\end{proof}
}

So if a black vertex added in turn $\ttime$ does not break \hc{}, then it has at least one neighbour who was alive in turn $\ttime-1$. The next two lemmas cover two cases. The first lemma states that if the new black vertex has at least two such neighbours, then no vertices die in turn $\ttime$. If, however, it has exactly one such neigbour, then some vertices die in turn $\ttime$ and the second lemma describes precisely which ones. 

\mythmspace{}

\begin{restatable}{mylemma}{lemifatleasttwoalive}
\begin{restatablelemma}[{{\rm Appendix~\refAppDead{} in \cite{DBLP:journals/corr/BosekLZS17}}}]
\label{lem:ifatleasttwoalive}
If $\black_{\ttime}$ has at least two neighbours which are alive in turn $\ttime-1$ then
$\black_{\ttime}$ is alive in turn $\ttime$ and no vertex dies in turn $\ttime$.
\end{restatablelemma}
\end{restatable}

\mythmspace{}

\MANUSCRIPTVERSION{
\begin{proof}
See Appendix~\refAppDead{}.
\end{proof}
}

The last lemma covers the only case when vertices die in turn $\ttime$. It shows that there is a certain region around $\black_{\ttime}$ where vertices die, and a barrier for that region are special vertices called life portals, defined below. The picture illustrating which region dies in turn $\ttime$ is given in Figure~\ref{fig:dead_parts}. 

\mythmspace{}

\begin{definition}\label{def:lifePortals}
A black vertex $\black$ is a life portal in turn $\ttime$ iff $\size{\neighbours{\ttime}{\black} \cap \Alive_{\ttime-1}} \geq 3$. The set of life portals in turn $\ttime$ is denoted as $\lp{\ttime}$.  
\end{definition}

\mythmspace{}

\begin{figure}[htbp]\fontsize{10}{10}\selectfont

\myfigspace{}

\centering \def\svgscale{0.9} \resizebox{0.6\textwidth}{!}{
\executeiffilenewer{dead_parts_2.svg}{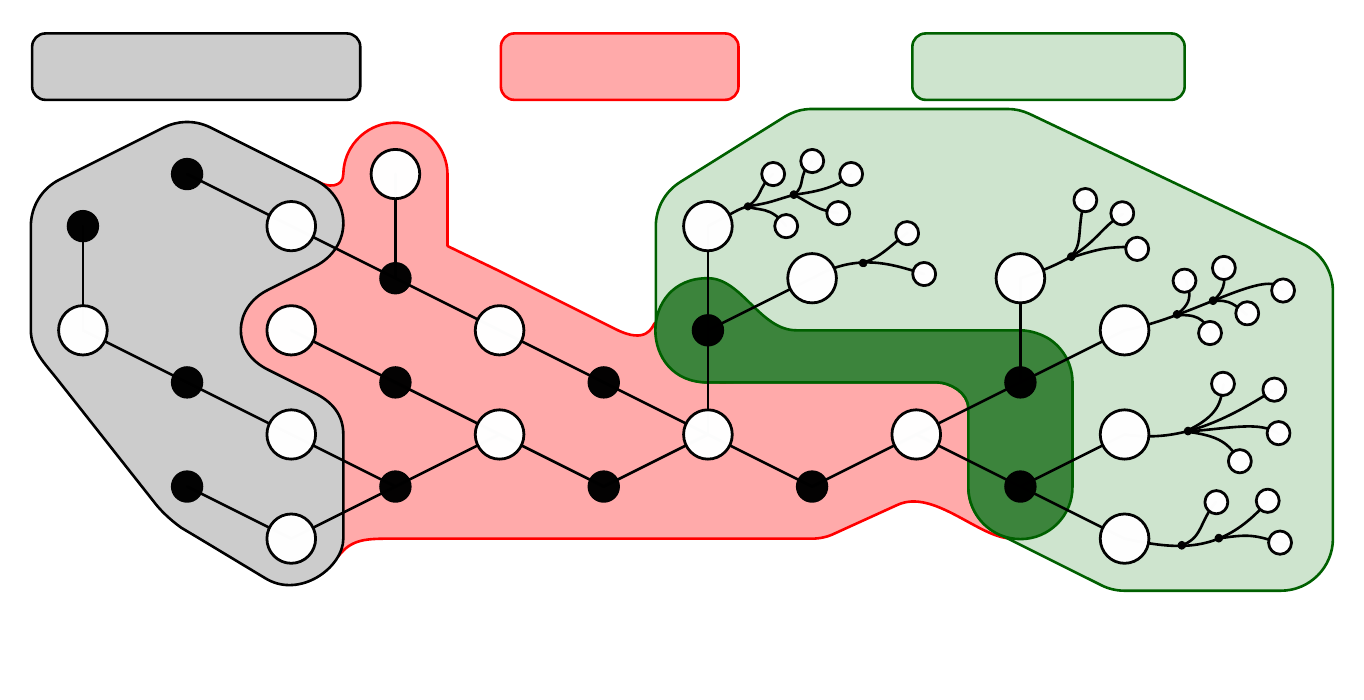} {inkscape -z -C --file=dead_parts_2.svg --export-pdf=dead_parts_2.pdf
--export-latex}
\begingroup
  \makeatletter
  \providecommand\color[2][]{
    \errmessage{(Inkscape) Color is used for the text in Inkscape, but the package 'color.sty' is not loaded}
    \renewcommand\color[2][]{}
  }
  \providecommand\transparent[1]{
    \errmessage{(Inkscape) Transparency is used (non-zero) for the text in Inkscape, but the package 'transparent.sty' is not loaded}
    \renewcommand\transparent[1]{}
  }
  \providecommand\rotatebox[2]{#2}
  \ifx\svgwidth\undefined
    \setlength{\unitlength}{392.76153324bp}
    \ifx\svgscale\undefined
      \relax
    \else
      \setlength{\unitlength}{\unitlength * \real{\svgscale}}
    \fi
  \else
    \setlength{\unitlength}{\svgwidth}
  \fi
  \global\let\svgwidth\undefined
  \global\let\svgscale\undefined
  \makeatother
  \begin{picture}(1,0.51305395)
    \put(0,0){\includegraphics[width=\unitlength,page=1]{dead_parts_2.pdf}}
    \put(0.37384202,0.45391129){\color[rgb]{0,0,0}\makebox(0,0)[lb]{\smash{Die in turn $t$}}}
    \put(0.030122,0.45391129){\color[rgb]{0,0,0}\makebox(0,0)[lb]{\smash{Dead in turn $t-1$}}}
    \put(0.67937092,0.45391129){\color[rgb]{0,0,0}\makebox(0,0)[lb]{\smash{Alive in turn $t$}}}
    \put(0,0){\includegraphics[width=\unitlength,page=2]{dead_parts_2.pdf}}
    \put(0.44258607,0.03380892){\color[rgb]{1,1,1}\makebox(0,0)[lb]{\smash{Life portals}}}
    \put(0,0){\includegraphics[width=\unitlength,page=3]{dead_parts_2.pdf}}
    \put(0.30635146,0.12967427){\color[rgb]{0,0,0}\makebox(0,0)[lb]{\smash{$b_t$}}}
  \end{picture}
\endgroup
  }
\caption{The vertices that die in turn $\ttime$.}
\label{fig:dead_parts}

\myfigspace{}

\end{figure}

\mythmspace{}

\begin{restatable}{mylemma}{lemdispatching}
\begin{restatablelemma}[{{\rm Appendix~\refAppDead{} in \cite{DBLP:journals/corr/BosekLZS17}}}]
\label{lem:dispatching}
\MANUSCRIPTVERSION{
If $\black_{\ttime}$ has exactly one neighbour in $\Forest_{\ttime}$ which was alive in turn $\ttime-1$ and there is a
path $\pathOne$ from $\black_{\ttime}$ to $v \in \Black_{\ttime} \cup \White$ such that
\begin{enumerate}
\item
all vertices of $\pathOne$ were alive in turn $\ttime-1$ and
\item
there are no life portals from $\lp{\ttime}$ on $\pathOne$,
\end{enumerate}
then $v$ dies in turn $\ttime$.
Vertices of $\Black_{\ttime} \cup \White$ that cannot be reached from $\black_{\ttime}$ via such path do not die in turn $\ttime$.
}
\LNCSVERSION{
If $\black_{\ttime}$ has exactly one neighbour in $\Forest_{\ttime}$ which was alive in turn $\ttime-1$ and there is a
path $\pathOne$ from $\black_{\ttime}$ to $v \in \Black_{\ttime} \cup \White$ such that
(1) all vertices of $\pathOne$ were alive in turn $\ttime-1$ and
(2) there are no life portals from $\lp{\ttime}$ on $\pathOne$,
then $v$ dies in turn $\ttime$.
Vertices of $\Black_{\ttime} \cup \White$ that cannot be reached from $\black_{\ttime}$ via such path do not die in turn $\ttime$.
}
\end{restatablelemma}
\end{restatable}

\mythmspace{}

\MANUSCRIPTVERSION{
\begin{proof}
See Appendix~\refAppDead{}.
\end{proof}
}

\mythmspace{}

\begin{mycorollary}\label{cor:23eqiv}
Lemma \ref{lem:dispatching} shows, that for all $\ttime$ such that $\black_{\ttime}$ does not break \hc{} and has exactly one neighbour alive in turn $\ttime-1$, statement $\size{\neighbours{\ttime}{b} \cap \Alive_{\ttime-1}} \geqslant 3$ in Definition
\ref{def:lifePortals} is equivalent to $\size{\neighbours{\ttime}{b} \cap \Alive_{\ttime}}
\geqslant 2$. 
\end{mycorollary}

\mythmspace{}
 \mysecspace{}

\section{The proof}\label{sec:proof}

\mysecspace{}

In the remainder of the paper we present the proof of Theorem \ref{thm:mainthm}, which states that if $\pathOne_{\ttime}$ is the path applied by \sap{} in turn $\ttime$, then $\sum_{\ttime=1}^n \length{ \pathOne_{\ttime} } \in \bigo{n \log n}$. 
As we mentioned in Section~\ref{sec:minimax}, we do not study $\length{ \pathOne_{\ttime} }$ directly.
Instead, we want to study distance functions $\dist{\ttime}{\black_{\ttime}}$ introduced in Section~\ref{sec:minimax}.
By Lemma~\ref{lem:pathineq} given in Section~\ref{sec:minimax} we know that $\dist{\ttime}{\black_{\ttime}}$ bounds $\length{ \pathOne_{\ttime} }$ from above.
Recall that $\dist{\ttime}{\black_{\ttime}}$ is the mini-max distance from $\black_{\ttime}$ to a white leaf in $\Forest_{\ttime}$.
By definition if $\dist{\ttime}{\black_{\ttime}} < \infinity$, then there is a path from $\black_{\ttime}$ to a white leaf which certifies it.
We introduce the formal definition of such path below.

\mythmspace{}

\begin{definition}
\label{def:path}
 Let $\ttime \in \range{n} $ and $v \in \Black_{\ttime} \cup \White$. Let $\RootedTreeOne$ be a connected component of $v$ in $\Forest_{\ttime}$ rooted at $v$. Then $\pathh{\ttime}{\firstVertex}=\mmpath{\RootedTreeOne}{v}$.
\end{definition}

\mythmspace{}

Note that by Definitions~\ref{def:minimax} and~\ref{def:distdir}, if $\dist{\ttime}{\black_{\ttime}} < \infinity$, then $\length{\pathh{\ttime}{\black_{\ttime}}}=\dist{\ttime}{\black_{\ttime}}$.
In addition to that, we define a path that certifies that $\secondDist{\ttime}{\black_{\ttime}}$ is finite.

\mythmspace{}

\begin{definition}\label{def:secPath}
 Let $\ttime \in \range{n} $ and $v \in \Black_{\ttime} \cup \White$. Let $\RootedTreeTwo$ be a connected component of $v$ in $\Forest_{\ttime} - \{ v, \dir{\ttime}{v} \}$ rooted at $v$, where $\Forest_{\ttime} - \{ v, \dir{\ttime}{v} \}$ denotes $\Forest_{\ttime}$ with edge $\{ v, \dir{\ttime}{v} \}$ removed. Then
 $\secondPath{\ttime}{\firstVertex}=\mmpath{\RootedTreeTwo}{v}$. 
\end{definition}

\mythmspace{}

Again by Definitions~\ref{def:minimax} and~\ref{def:distdir}, if $\secondDist{\ttime}{\black_{\ttime}} < \infinity$, then $\length{\secondPath{\ttime}{\black_{\ttime}}}=\secondDist{\ttime}{\black_{\ttime}}$.
Instead of proving Theorem~\ref{thm:mainthm}, in the remainder of this paper we prove that $\sum_{\ttime: \dist{\ttime}{\black_{\ttime}}<\infinity}\length{\pathh{\ttime}{\black_{\ttime}}} \in \bigo{ n \log n }$. This is in fact a stronger statement. We claim that even if the adversary picks the worst possible maximum matching in each turn, the \sap{} still applies paths of total length $\bigo{ n \log n }$. Note that if $\dist{\ttime}{\black_{\ttime}} = \infinity$, then due to Corollary~\ref{cor:nonMatch} \sap{} cannot match the new vertex $\black_{\ttime}$ if the maximum matching is given on the remaining vertices. 
Our proof of such simple statement is unfortunately rather complex.
Before we move on to it, we give some intuitions on where the actual problem hides.
It is enlightening to discover, that with the additional assumption that the black vertices are of
degree two or more, the statement above has a very simple proof.

\mythmspace{}

\begin{mylemma}
\label{lem:deggeq2}
If each black vertex $\black_\ttime$ has degree at least $2$, then
$\sum_{\ttime=1}^n \length{\pathh{\ttime}{\black_\ttime}} \leqslant n \log_2n$.
\end{mylemma}

\mythmspace{}

\mythmspace{}

\begin{proof}
We start by observing that no vertex ever dies. In turn $\ttime=0$ the only presented vertices are $\White$, so by definition all vertices are alive in turn $\ttime=0$. Let $\ttime>0$ and assume that no vertices died up until turn $\ttime-1$. Due to Lemma~\ref{lem:ifatleasttwoalive} no vertex dies in turn $\ttime$ and $\black_{\ttime}$ is alive in turn $\ttime$. This implies that $\dist{\ttime}{\black_{\ttime}} < \infinity$ and $\secondDist{\ttime}{\black_{\ttime}} < \infinity$.
Hence, $\pathh{\ttime}{\black_{\ttime}}$ and $\secondPath{\ttime}{\black_{\ttime}}$ are two separate paths, contained in two different components of $\Forest_{\ttime-1}$ connected in turn $\ttime$ by $\black_{\ttime}$. Also, $\length{\pathh{\ttime}{\black_{\ttime}}} \leq \length{\secondPath{\ttime}{\black_{\ttime}}}$.
Thus, $\pathh{\ttime}{\black_{\ttime}}$ is at most as long as the size of the smaller of the two components.
We pay for $\pathh{\ttime}{\black_\ttime}$ by charging $1$ token to each vertex in every component but the largest one among the components of $\Forest_{\ttime-1}$ connected by $\black_{\ttime}$ in turn $\ttime$.
A vertex $\firstVertex$ is charged when $\firstVertex$'s component size increases at least twice,
so $\firstVertex$ cannot be charged more than $\log_2n$ times. This gives a total charge of at most $n \log_2n$.
\LNCSVERSION{\qed}\end{proof}

\mythmspace{}

The essence of this proof is that every time a black vertex is added, it connects
at least two trees into one.
As a consequence there are at least two alternative mini-max paths starting from the
added vertex, each in a separate tree.
The length of the shorter of the two can be charged to the vertices of the smaller tree.
If we allow black vertices of degree $1$, the situation becomes more complicated, because:
(1) there is no alternative path, i.e., the path needs to follow the only edge adjacent to the newly
added black vertex, and (2) no trees are merged.
Nevertheless the proof of Theorem \ref{thm:mainthm} is a generalization of the proof of Lemma~\ref{lem:deggeq2}. The majority of the remainder of this paper is devoted to addressing issue (2). We define trees which are merged in each turn and allow introducing the charging scheme that generalizes the scheme of Lemma~\ref{lem:deggeq2}. We start, though, by addressing issue (1). To that end we introduce a concept of a dispatching vertex.
Even though $\black_\ttime$ does not necessarily fork into two alternative mini-max paths, there
is another vertex which does. It is the first life portal on $\pathh{\ttime}{\black_{\ttime}}$.
We refer to it as dispatching vertex in turn $\ttime$. To be more formal, we introduce the following definition.

\mythmspace{}

\begin{definition}
\label{def:dispatching}
The dispatching vertex at time $\ttime \in \range{n} $ is the first black vertex on $\pathh{\ttime}{\black_{\ttime}}$ such that $\size{\neighbours{\ttime}{\black} \cap \Alive_{\ttime}} \geqslant 2$. We denote it as $\disp{\ttime}$.
\end{definition}

\mythmspace{}

First observe that $\disp{\ttime}$ has two alive neighbours in turn $\ttime$, so there are two alternative mini-max paths branching from $\disp{\ttime}$. 
Our next observation is that if $\disp{\ttime}\neq \black_{\ttime}$, then $\disp{\ttime} \in \lp{\ttime}$: if $\black_{\ttime}$ has two neighbours alive in $\ttime-1$, then due to Lemma~\ref{lem:ifatleasttwoalive} no vertex dies in turn $\ttime$ so $\black_{\ttime}$ has two neigbours alive in turn $\ttime$ and hence $\black_{\ttime}=\disp{\ttime}$; otherwise, if $\black_{\ttime}$ has one neigbour alive in turn $\ttime-1$, then due to Lemma~\ref{lem:dispatching} and Corollary~\ref{cor:23eqiv} it holds that $\disp{\ttime}$ is the first life portal of $\lp{\ttime}$ on $\pathh{\ttime}{\black_{\ttime}}$. Then also $\disp{\ttime}$ is the first vertex on $\pathh{\ttime}{\black_{\ttime}}$ that remains alive. All vertices that follow $\disp{\ttime}$ on $\pathh{\ttime}{\black_{\ttime}}$  remain alive as well. It may happen that $\pathh{\ttime}{\black_{\ttime}}$ contains no life portals, in which case there is no dispatching vertex defined in turn $\ttime$. This case however is not of concern, since the whole $\pathh{\ttime}{\black_{\ttime}}$ dies in turn $\ttime$ due to Lemma~\ref{lem:dispatching}. In general, we do not have to worry about vertices that die, and we state this observation as Observation~\ref{obs:dieOnlyOnce} preceded by Definition~\ref{def:presuf}.

\mythmspace{}

\begin{definition}\label{def:presuf}
Let $\ttime \in \range{n}$ be such that $\dist{\ttime}{\black_{\ttime}} < \infinity$. We let $\pathh{\ttime}{\black_{\ttime}}=\pathhp{\ttime}{\black_{\ttime}} \cdot \pathhs{\ttime}{\black_{\ttime}}$, where $\pathhp{\ttime}{\black_{\ttime}}$ is the prefix of $\pathh{\ttime}{\black_{\ttime}}$ that dies (possibly empty) and $\pathhs{\ttime}{\black_{\ttime}}$ is the corresponding suffix (also possibly empty).  
\end{definition}

\mythmspace{}

Note that if $\disp{\ttime}$ is defined then $\pathhs{\ttime}{\black_{\ttime}}$ begins with $\disp{\ttime}$.
Since the final forest has $2n$ vertices and each can die only once, we have the following: 

\mythmspace{}

\begin{myobservation}\label{obs:dieOnlyOnce}
$\sum_{\ttime:\dist{\ttime}{\black_{\ttime}}<\infinity}\length{\pathhp{\ttime}{\black_{\ttime}}} \leq 2n$.
\end{myobservation}

\mythmspace{}

Thus, to bound $\sum_{\ttime:\dist{\ttime}{\black_{\ttime}}<\infinity}\length{\pathh{\ttime}{\black_{\ttime}}}$ it suffices to bound \LNCSVERSION{\linebreak} $\sum_{\ttime:\dist{\ttime}{\black_{\ttime}}<\infinity}\length{\pathhs{\ttime}{\black_{\ttime}}}$.
We conclude the list of properties of the dispatching vertex with the following observation.

\mythmspace{}

\begin{restatable}{myobservation}{obsdistPathDisp}
\begin{restatableobservation}[{{\rm Appendix~\refAppProof{} in \cite{DBLP:journals/corr/BosekLZS17}}}]
\label{obs:distPathDisp}
Let $\ttime \in \range{n}$ and $\dist{\ttime}{\black_{\ttime}} < \infinity$ and $\disp{\ttime}$ is defined. Then $\length{ \pathhs{\ttime}{\black_{\ttime}} }=\dist{\ttime}{\disp{\ttime}}$.
\end{restatableobservation}
\end{restatable}

\mythmspace{}

\MANUSCRIPTVERSION{
\begin{proof}
See Appendix~\refAppProof{}.
\end{proof}
}

To proceed further, we introduce the crucial notion in our proof: the notion of a level.
The levels are some numbers assigned to vertices: each vertex is assigned its level.
The intuitive meaning of the level of a vertex is the following.
Consider $\pathh{\ttime}{\black_\ttime}$, which is the worst case shortest augmenting path for
a black vertex $\black_\ttime$ added in turn $\ttime$.
For a vertex $v$, if $\pathh{\ttime}{\black_\ttime}$ crosses $v$, level in $\Forest_{\ttime}$
returns the value representing the length of the suffix of $\pathh{\ttime}{\black_\ttime}$
starting in $v$.
Formally, the level function is defined in the following way.

\mythmspace{}

\begin{definition}
\label{def:level}
\MANUSCRIPTVERSION{
For $\firstVertex \in \White \cup \Black$ and $\ttime \in \set{0, \ldots,n}$ let
$$
\level{\ttime}{ \firstVertex}= \left \lbrace
\begin{array}
{ll} \secondDist{\ttime}{\firstVertex} & \firstVertex \in \White, \\
\dist{\ttime}{\firstVertex} & \firstVertex \in \Black_{\ttime}, \\
0 & \text{otherwise.
}
\end{array}
\right.
$$
}
\LNCSVERSION{
For $\firstVertex \in \White \cup \Black$ and $\ttime \in \set{0, \ldots,n}$ let
$\level{\ttime}{ \firstVertex}= \secondDist{\ttime}{\firstVertex}$ if $\firstVertex \in \White$, 
$\level{\ttime}{ \firstVertex}= \dist{\ttime}{\firstVertex}$ if $\firstVertex \in \Black_{\ttime}$, and
$\level{\ttime}{ \firstVertex}=0$ otherwise.
}
\end{definition}

\mythmspace{}

It may at first seem confusing that the level of a white vertex is the second maximum distance to a
leaf.
It is defined this way because, surprisingly, in every turn $\ttime$ the path
$\pathh{\ttime}{\black_\ttime}$ enters its white vertices through the edge determining the
maximum distance from the white vertex to a leaf.
We illustrate the introduced definitions in Figure~\ref{fig:levels}.  
We present there an example run of an online scenario together with changing levels of vertices.
We mark the dispatching vertices in each turn. 
An important property of the level function is that the levels of vertices drop by at most one along both $\pathh{\ttime}{\black_\ttime}$ and $\secondPath{\ttime}{\black_{\ttime}}$. 
\begin{figure}[htbp]\fontsize{10}{10}\selectfont

\myfigspace{}

  \centering \def\svgscale{0.9} \resizebox{0.7\textwidth}{!}{
\executeiffilenewer{levels.svg}{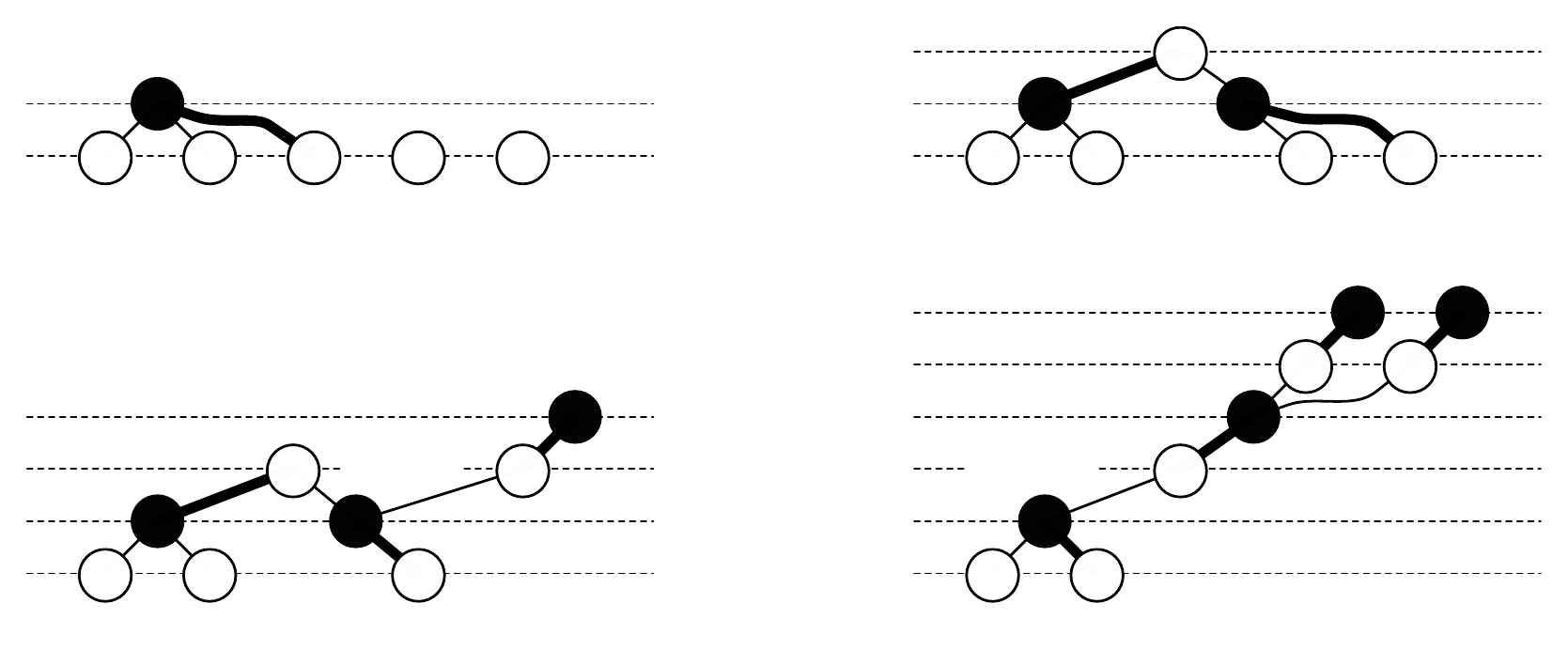} {inkscape -z -C --file=levels.svg --export-pdf=levels.pdf
--export-latex}
\begingroup
  \makeatletter
  \providecommand\color[2][]{
    \errmessage{(Inkscape) Color is used for the text in Inkscape, but the package 'color.sty' is not loaded}
    \renewcommand\color[2][]{}
  }
  \providecommand\transparent[1]{
    \errmessage{(Inkscape) Transparency is used (non-zero) for the text in Inkscape, but the package 'transparent.sty' is not loaded}
    \renewcommand\transparent[1]{}
  }
  \providecommand\rotatebox[2]{#2}
  \ifx\svgwidth\undefined
    \setlength{\unitlength}{480.56794739bp}
    \ifx\svgscale\undefined
      \relax
    \else
      \setlength{\unitlength}{\unitlength * \real{\svgscale}}
    \fi
  \else
    \setlength{\unitlength}{\svgwidth}
  \fi
  \global\let\svgwidth\undefined
  \global\let\svgscale\undefined
  \makeatother
  \begin{picture}(1,0.42410928)
    \put(0,0){\includegraphics[width=\unitlength,page=1]{levels.pdf}}
    \put(0.46504138,0.31792541){\color[rgb]{0,0,0}\makebox(0,0)[lb]{\smash{level 0}}}
    \put(0.46504138,0.35121935){\color[rgb]{0,0,0}\makebox(0,0)[lb]{\smash{level 1}}}
    \put(0.05719063,0.28796089){\color[rgb]{0,0,0}\makebox(0,0)[lb]{\smash{$w_1$}}}
    \put(0.12377852,0.28796089){\color[rgb]{0,0,0}\makebox(0,0)[lb]{\smash{$w_2$}}}
    \put(0.19036638,0.28796089){\color[rgb]{0,0,0}\makebox(0,0)[lb]{\smash{$w_3$}}}
    \put(0.25695425,0.28796089){\color[rgb]{0,0,0}\makebox(0,0)[lb]{\smash{$w_4$}}}
    \put(0.32354216,0.28796089){\color[rgb]{0,0,0}\makebox(0,0)[lb]{\smash{$w_5$}}}
    \put(0.12458549,0.36786632){\color[rgb]{0,0,0}\makebox(0,0)[lb]{\smash{$b_1=\check{b}_1$}}}
    \put(0.62318761,0.28796089){\color[rgb]{0,0,0}\makebox(0,0)[lb]{\smash{$w_1$}}}
    \put(0.68977547,0.28796089){\color[rgb]{0,0,0}\makebox(0,0)[lb]{\smash{$w_2$}}}
    \put(0.75636338,0.28796089){\color[rgb]{0,0,0}\makebox(0,0)[lb]{\smash{$w_3$}}}
    \put(0.8229512,0.28796089){\color[rgb]{0,0,0}\makebox(0,0)[lb]{\smash{$w_4$}}}
    \put(0.88953911,0.28796089){\color[rgb]{0,0,0}\makebox(0,0)[lb]{\smash{$w_5$}}}
    \put(0.81349012,0.36572599){\color[rgb]{0,0,0}\makebox(0,0)[lb]{\smash{$b_2=\check{b}_2$}}}
    \put(0.05719063,0.02160935){\color[rgb]{0,0,0}\makebox(0,0)[lb]{\smash{$w_1$}}}
    \put(0.12377852,0.02160935){\color[rgb]{0,0,0}\makebox(0,0)[lb]{\smash{$w_2$}}}
    \put(0.19036638,0.02160935){\color[rgb]{0,0,0}\makebox(0,0)[lb]{\smash{$w_3$}}}
    \put(0.25695425,0.02160935){\color[rgb]{0,0,0}\makebox(0,0)[lb]{\smash{$w_4$}}}
    \put(0.32354216,0.02160935){\color[rgb]{0,0,0}\makebox(0,0)[lb]{\smash{$w_5$}}}
    \put(0.22068763,0.11845902){\color[rgb]{0,0,0}\makebox(0,0)[lb]{\smash{$b_2=\check{b}_3$}}}
    \put(0.46504138,0.05157391){\color[rgb]{0,0,0}\makebox(0,0)[lb]{\smash{level 0}}}
    \put(0.46504138,0.08486782){\color[rgb]{0,0,0}\makebox(0,0)[lb]{\smash{level 1}}}
    \put(0.46504138,0.11816182){\color[rgb]{0,0,0}\makebox(0,0)[lb]{\smash{level 2}}}
    \put(0.46504138,0.15145573){\color[rgb]{0,0,0}\makebox(0,0)[lb]{\smash{level 3}}}
    \put(0.62318761,0.02160935){\color[rgb]{0,0,0}\makebox(0,0)[lb]{\smash{$w_1$}}}
    \put(0.68977547,0.02160935){\color[rgb]{0,0,0}\makebox(0,0)[lb]{\smash{$w_2$}}}
    \put(0.75636338,0.02160935){\color[rgb]{0,0,0}\makebox(0,0)[lb]{\smash{$w_3$}}}
    \put(0.8229512,0.02160935){\color[rgb]{0,0,0}\makebox(0,0)[lb]{\smash{$w_4$}}}
    \put(0.88953911,0.02160935){\color[rgb]{0,0,0}\makebox(0,0)[lb]{\smash{$w_5$}}}
    \put(0.62363716,0.11754228){\color[rgb]{0,0,0}\makebox(0,0)[lb]{\smash{$b_1=\check{b}_4$}}}
    \put(0.46504138,0.18474964){\color[rgb]{0,0,0}\makebox(0,0)[lb]{\smash{level 4}}}
    \put(0.46504138,0.21804364){\color[rgb]{0,0,0}\makebox(0,0)[lb]{\smash{level 5}}}
    \put(0.76064422,0.16620868){\color[rgb]{0,0,0}\makebox(0,0)[lb]{\smash{$b_2$}}}
    \put(0.82628063,0.23388355){\color[rgb]{0,0,0}\makebox(0,0)[lb]{\smash{$b_4$}}}
    \put(0.89286854,0.23388355){\color[rgb]{0,0,0}\makebox(0,0)[lb]{\smash{$b_3$}}}
    \put(0.62711153,0.36548819){\color[rgb]{0,0,0}\makebox(0,0)[lb]{\smash{$b_1$}}}
    \put(0.06069417,0.09877996){\color[rgb]{0,0,0}\makebox(0,0)[lb]{\smash{$b_1$}}}
    \put(0.32687155,0.18142029){\color[rgb]{0,0,0}\makebox(0,0)[lb]{\smash{$b_3$}}}
    \put(0.46504138,0.38451327){\color[rgb]{0,0,0}\makebox(0,0)[lb]{\smash{level 2}}}
  \end{picture}
\endgroup
  }
  \caption{Levels}
  \label{fig:levels}

\myfigspace{}

  \end{figure}

\mythmspace{}

\begin{restatable}{mylemma}{lemplusMinusOne}
\begin{restatablelemma}[{{\rm Appendix~\refAppProof{} in \cite{DBLP:journals/corr/BosekLZS17}}}]
\label{lem:plusMinusOne}
For $\firstVertex\in\White\cup\Black_{\ttime}$ and $\secondVertex\in\set{\dir{\ttime}{\firstVertex},\secondDir{\ttime}{\firstVertex}}$ it holds that $\abs{\level{\ttime}{\firstVertex}-\level{\ttime}{\secondVertex}}\leqslant 1$.
\end{restatablelemma}
\end{restatable}

\mythmspace{}

\MANUSCRIPTVERSION{
\begin{proof}
See Appendix~\refAppProof{}.
\end{proof}
}

We are ready to move on to the proof of Theorem \ref{thm:mainthm}.
We consider two cases:
\begin{enumerate}
\item
\label{enu:easy case}
the level of a dispatching vertex in turn $\ttime$ grows by at most a factor of $\constDist$
\item
\label{enu:hard case}
the level of a dispatching vertex in turn $\ttime$ grows by more than a factor of $\constDist$
\end{enumerate}
where $\constDist$ is some constant value greater than $1$ which we reveal later on.
The total length of paths $\pathhs{\ttime}{\black_{\ttime}}$ satisfying case (\ref{enu:easy case}) is bounded by Lemma \ref{lem:easy case new}
while the total length of paths $\pathhs{\ttime}{\black_{\ttime}}$ satisfying case (\ref{enu:hard case}) is bounded by Lemma \ref{lem:hard case}.  
 
\mythmspace{}

\begin{mylemma}
\label{lem:easy case new}For cases when $\dist{\ttime}{\black_{\ttime}}<\infinity$, $\disp{\ttime}$ is defined and
$\level{\ttime}{\disp{\ttime}}< \constDist \level{ \ttime-1}{\disp{\ttime}}$
the total length of paths $\pathhs{\ttime}{\black_{\ttime}}$  is bounded by
$2 \constDist n + \constDist n \log_2 n$.
\end{mylemma}

\mythmspace{}

\mythmspace{}

\begin{proof}
Let $\ttime$ be such that it satisfies the assumptions of the lemma.
First observe that $\black_{\ttime} \neq \disp{\ttime}$, otherwise $\level{\ttime-1}{\disp{\ttime}}=0 < \level{\ttime}{\disp{\ttime}}/\constDist$.
Due to Lemma~\ref{lem:aliveneighbour} and Lemma~\ref{lem:ifatleasttwoalive}, $\black_{\ttime}$ has precisely one neighbour alive in turn $\ttime-1$.

In order to show an appropriate charging scheme, consider the final forest $\Forest=\Forest_n$.
We study the connected components of a subforest $\InducedTree{\Alive_{\ttime}}$ of $\Forest$ induced on vertices alive in turn $\ttime$.
Recall that vertices not yet presented are considered alive.
In turn $\ttime$ some vertices, in particular $\black_{\ttime}$, die.
Due to Lemma~\ref{lem:dispatching} vertices that die in turn $\ttime$ form a connected component $D$ of $\InducedTree{\Alive_{\ttime-1}}$.
Then the connected component $C$ of $\disp{\ttime}$ in $\InducedTree{\Alive_{\ttime-1}}$ splits into $D$ and components $C_1,\ldots,C_k$ in $\InducedTree{\Alive_{\ttime}}$. 
Let $C_1$ be the largest component among $C_1,\ldots,C_k$.
Due to Lemma~\ref{lem:dispatching}, $\pathhs{\ttime}{\black_{\ttime}}$ is contained entirely in one of the components $C_1,\ldots,C_k$, say $\pathhs{\ttime}{\black_{\ttime}}$ is contained in $C_i$.
If $i \neq 1$, we can charge the length of $\pathhs{\ttime}{\black_{\ttime}}$ by charging $1$ token to the vertices of $C_i$.
A particular vertex can be charged at most $\log_2 n$ tokens this way, as each time it is charged its component halves the size.
It remains to deal with the case when $\pathhs{\ttime}{\black_{\ttime}}$ is contained in the largest component $C_1$.
For the reference see Figure~\ref{pic:easy case}.
\begin{figure}[htbp]\fontsize{10}{10}\selectfont

\myfigspace{}

\centering \def\svgscale{0.9} \resizebox{0.5\textwidth}{!}{
\executeiffilenewer{smallinc.svg}{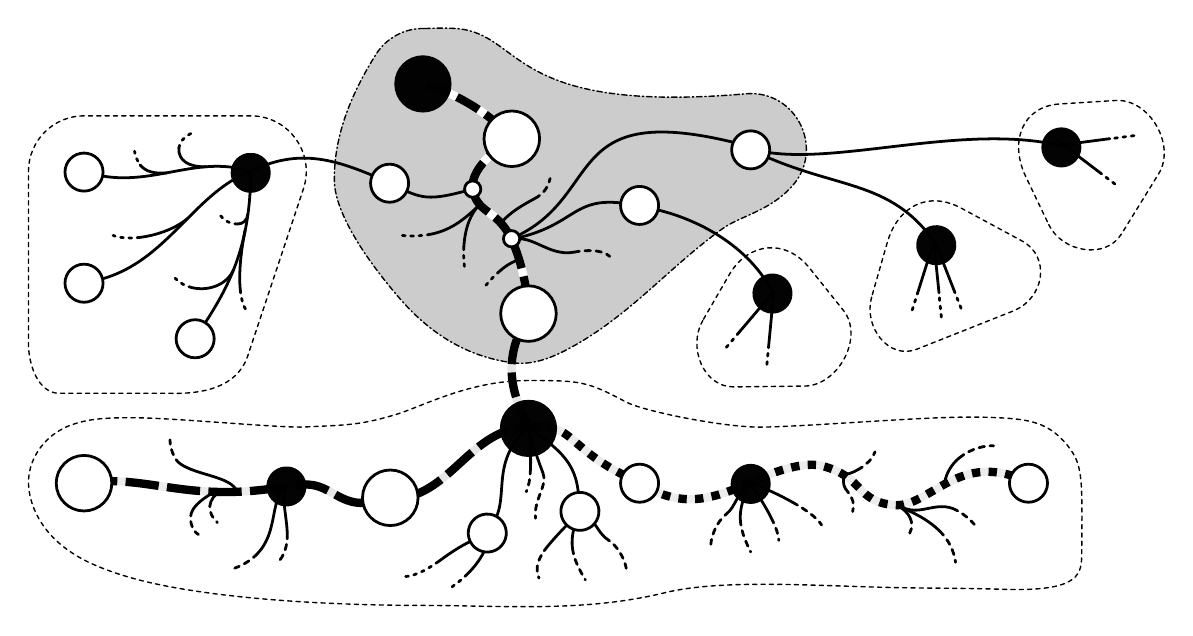} {inkscape -z -C --file=smallinc.svg --export-pdf=smallinc.pdf
--export-latex}
\begingroup
  \makeatletter
  \providecommand\color[2][]{
    \errmessage{(Inkscape) Color is used for the text in Inkscape, but the package 'color.sty' is not loaded}
    \renewcommand\color[2][]{}
  }
  \providecommand\transparent[1]{
    \errmessage{(Inkscape) Transparency is used (non-zero) for the text in Inkscape, but the package 'transparent.sty' is not loaded}
    \renewcommand\transparent[1]{}
  }
  \providecommand\rotatebox[2]{#2}
  \ifx\svgwidth\undefined
    \setlength{\unitlength}{343.49954224bp}
    \ifx\svgscale\undefined
      \relax
    \else
      \setlength{\unitlength}{\unitlength * \real{\svgscale}}
    \fi
  \else
    \setlength{\unitlength}{\svgwidth}
  \fi
  \global\let\svgwidth\undefined
  \global\let\svgscale\undefined
  \makeatother
  \begin{picture}(1,0.53253305)
    \put(0,0){\includegraphics[width=\unitlength,page=1]{smallinc.pdf}}
    \put(0.852987,0.06046234){\color[rgb]{0,0,0}\makebox(0,0)[lb]{\smash{$C_1$}}}
    \put(0.05182138,0.22814812){\color[rgb]{0,0,0}\makebox(0,0)[lb]{\smash{$C_2$}}}
    \put(0.66036465,0.23445281){\color[rgb]{0,0,0}\makebox(0,0)[lb]{\smash{$C_3$}}}
    \put(0.82503926,0.29335921){\color[rgb]{0,0,0}\makebox(0,0)[lb]{\smash{$C_4$}}}
    \put(0.89083271,0.34488768){\color[rgb]{0,0,0}\makebox(0,0)[lb]{\smash{$C_5$}}}
    \put(0.38340847,0.46774083){\color[rgb]{0,0,0}\makebox(0,0)[lb]{\smash{$b_t$}}}
    \put(0.47540276,0.17487289){\color[rgb]{0,0,0}\makebox(0,0)[lb]{\smash{$\disp{\ttime}$}}}
    \put(0.47336488,0.27065177){\color[rgb]{0,0,0}\makebox(0,0)[lb]{\smash{$\white^p$}}}
  \end{picture}
\endgroup
  }
\caption{Splitting the component $C$ into $C_1,\ldots,C_5$ and $D$.}
\label{pic:easy case}

\myfigspace{}

\end{figure}
Let $\white^p$ be the predecessor of $\disp{\ttime}$ on $\pathh{\ttime}{\black_{\ttime}}$.
Let $\RootedTreeOne$ be the connected component of $\disp{\ttime}$ in $\Forest_{\ttime-1}$ rooted at $\disp{\ttime}$ and let $\RootedTreeROne$ be the connected component of $\disp{\ttime}$ in $\Forest_{\ttime-1}$ rooted at $\white^p$.
For the reference see Figure~\ref{pic:easy case 1}.
\begin{figure}[htbp]\fontsize{10}{10}\selectfont

\myfigspace{}

\centering \def\svgscale{0.9} \resizebox{0.7\textwidth}{!}{
\executeiffilenewer{easy_case_1-popr.svg}{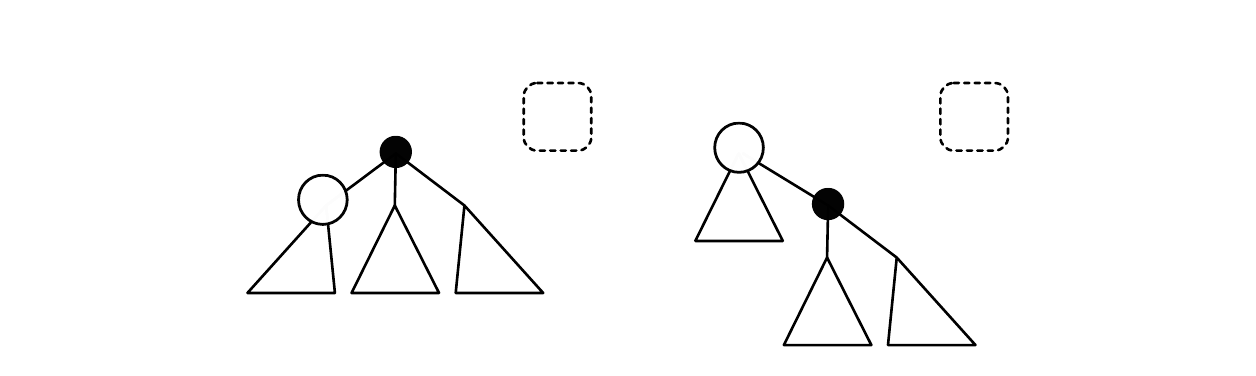} {inkscape -z -C --file=easy_case_1-popr.svg --export-pdf=easy_case_1-popr.pdf
--export-latex}
\begingroup
  \makeatletter
  \providecommand\color[2][]{
    \errmessage{(Inkscape) Color is used for the text in Inkscape, but the package 'color.sty' is not loaded}
    \renewcommand\color[2][]{}
  }
  \providecommand\transparent[1]{
    \errmessage{(Inkscape) Transparency is used (non-zero) for the text in Inkscape, but the package 'transparent.sty' is not loaded}
    \renewcommand\transparent[1]{}
  }
  \providecommand\rotatebox[2]{#2}
  \ifx\svgwidth\undefined
    \setlength{\unitlength}{361.53280904bp}
    \ifx\svgscale\undefined
      \relax
    \else
      \setlength{\unitlength}{\unitlength * \real{\svgscale}}
    \fi
  \else
    \setlength{\unitlength}{\svgwidth}
  \fi
  \global\let\svgwidth\undefined
  \global\let\svgscale\undefined
  \makeatother
  \begin{picture}(1,0.29950054)
    \put(0,0){\includegraphics[width=\unitlength,page=1]{easy_case_1-popr.pdf}}
    \put(0.30788684,0.20231764){\color[rgb]{0,0,0}\makebox(0,0)[lb]{\smash{$\disp{\ttime}$}}}
    \put(0.30226969,0.07784737){\color[rgb]{0,0,0}\makebox(0,0)[lb]{\smash{$B$}}}
    \put(0.37280296,0.07784737){\color[rgb]{0,0,0}\makebox(0,0)[lb]{\smash{$C$}}}
    \put(0.23173762,0.07784737){\color[rgb]{0,0,0}\makebox(0,0)[lb]{\smash{$A$}}}
    \put(0.2207572,0.16290296){\color[rgb]{0,0,0}\makebox(0,0)[lb]{\smash{$\white^p$}}}
    \put(0.65225468,0.16082771){\color[rgb]{0,0,0}\makebox(0,0)[lb]{\smash{$\disp{\ttime}$}}}
    \put(0.64663753,0.03635744){\color[rgb]{0,0,0}\makebox(0,0)[lb]{\smash{$B$}}}
    \put(0.71717079,0.03635744){\color[rgb]{0,0,0}\makebox(0,0)[lb]{\smash{$C$}}}
    \put(0.57610486,0.1193373){\color[rgb]{0,0,0}\makebox(0,0)[lb]{\smash{$A$}}}
    \put(0.5526776,0.20439289){\color[rgb]{0,0,0}\makebox(0,0)[lb]{\smash{$\white^p$}}}
    \put(0.42924232,0.19401951){\color[rgb]{0,0,0}\makebox(0,0)[lb]{\smash{$\RootedTreeOne$}}}
    \put(0.75597651,0.19401951){\color[rgb]{0,0,0}\makebox(0,0)[lb]{\smash{$\RootedTreeROne$}}}
    \put(0.20416142,0.25003117){\color[rgb]{0,0,0}\makebox(0,0)[lb]{\smash{root $\disp{\ttime}$:}}}
    \put(0.56512503,0.25003117){\color[rgb]{0,0,0}\makebox(0,0)[lb]{\smash{root $w^p$:}}}
  \end{picture}
\endgroup
  }
\caption{The connected component of $\Forest_{\ttime-1}$ rooted in $\disp{\ttime}$ and $\white^{p}$.}
\label{pic:easy case 1}

\myfigspace{}

\end{figure}
By Observation~\ref{obs:distPathDisp} and our assumptions it holds that
\MANUSCRIPTVERSION{
\begin{multline*}
\length{\pathhs{\ttime}{\black_{\ttime}}}=
\dist{\ttime}{\disp{\ttime}}<
\constDist\dist{\ttime-1}{\disp{\ttime}}=\\
=
\constDist\min_{\white \in \children{\RootedTreeOne}{\disp{\ttime}}} \treeDist{\RootedTreeOne}{\white} + \constDist \leqslant
\constDist\treeDist{\RootedTreeOne}{\white^p} +\constDist.
\end{multline*}
}
\LNCSVERSION{
$
\length{\pathhs{\ttime}{\black_{\ttime}}}=
\dist{\ttime}{\disp{\ttime}}<
\constDist\dist{\ttime-1}{\disp{\ttime}}
=
\constDist\min_{\white \in \children{\RootedTreeOne}{\disp{\ttime}}} \treeDist{\RootedTreeOne}{\white} + \constDist \leqslant
\constDist\treeDist{\RootedTreeOne}{\white^p} +\constDist.
$
}
The constant cost of $\constDist$ gives a total cost of $\constDist n$ over all turns.
What remains to show is how to charge the cost of $\constDist\treeDist{\RootedTreeOne}{\white^p}$.
Since $\white^p\in \Alive_{\ttime-1}$ it holds that
\MANUSCRIPTVERSION{
\begin{multline*}
\treeDist{\RootedTreeOne}{\white^p}=
\max_{\black \in \children{\RootedTreeROne}{\white^p}\setminus\set{\disp{\ttime}}}\treeDist{\RootedTreeROne}{\black}+1\leqslant\\
\leqslant \max_{\black \in \children{\RootedTreeROne}{\white^p} } \treeDist{\RootedTreeROne}{\black} + 1=
\dist{\ttime-1}{\white^p}<
\infinity.
\end{multline*}
}
\LNCSVERSION{
$
\treeDist{\RootedTreeOne}{\white^p}=
\max_{\black \in \children{\RootedTreeROne}{\white^p}\setminus\set{\disp{\ttime}}}\treeDist{\RootedTreeROne}{\black}+1
\leqslant \max_{\black \in \children{\RootedTreeROne}{\white^p} } \treeDist{\RootedTreeROne}{\black} + 1=
\dist{\ttime-1}{\white^p}<
\infinity.
$
}
Thus $\treeDist{\RootedTreeOne}{\white^p}=\length{\mmpath{\RootedTreeOne}{\white^p}}$. We charge the vertices of $\mmpath{\RootedTreeOne}{\white^p}$ to pay for the cost given by its length.
Due to Lemma~\ref{lem:mmPathsAlive} it holds that $\mmpath{\RootedTreeOne}{\white^p}$ visits only vertices that are alive in turn $\ttime-1$.
By definition $V(\mmpath{\RootedTreeOne}{\white^p}) \cap C_1 = \emptyset$, so each vertex of $\mmpath{\RootedTreeOne}{\white^p}$ either dies in turn $\ttime$ or is contained in $C_i$ for $i>1$.
To pay for that, we charge $\constDist$ tokens to every vertex that dies in turn $\ttime$ and we charge $\constDist$ tokens to each vertex of components $C_2 \ldots C_k$.
The total charge for this case sums up to $\constDist n + \constDist n \log_2 n$.
If we add the charge we needed for other cases, we obtain a total of $2 \constDist n + \constDist n \log_2 n$.
\LNCSVERSION{\qed}\end{proof}

\mythmspace{}
 
\mythmspace{}

\begin{mylemma}
\label{lem:hard case}
For cases when $\dist{\ttime}{\black_{\ttime}} < \infinity$, $\disp{\ttime}$ is defined and $\level{\ttime}{\disp{\ttime}} \geqslant \constDist
\level{\ttime-1}{\disp{\ttime}}$ the sum of the lengths of paths $\pathhs{\ttime}{\black_{\ttime}}$ is bounded by
$\frac{\constDist\bra{\constDist+1}}{\bra{\constDist-1}^2} n (2 \ln n +3.4) +n$.
\end{mylemma}

\mythmspace{}

\mythmspace{}

\begin{proof}
Given the level function, we want to consider the vertices of $\Forest=\Forest_{n}$ whose level in turn
$\ttime$ is above a certain value $\lev$.
To be more precise, we need to consider the subforest of $\Forest$ induced by such vertices.
This forest changes dynamically as the turns pass by.
We describe it more formally below.

\mythmspace{}

\begin{definition}
For $t\in \range{n}$ and $l \in \mathbb{N}$ we define
$\TreeLevell{\ttime}{l}= \InducedTree{\set{\firstVertex \in \White \cup \Black:
l \leqslant \level{\ttime}{\firstVertex}}}$.
\end{definition}

\mythmspace{}

Recall that if $\black \in \Black \setminus \Black_{\ttime}$ then $\level{\ttime}{\black}=0$.
For a subforest $\Forest' = \tuple{\Vertices', \Edges'}$ of $\Forest =
\tuple{\White \cup \Black, \Edges}$, we denote a connected component of vertex $\firstVertex
\in \Vertices'$ as $\comp{\firstVertex}{\Forest'}$.
The family of all connected components is denoted as $\Comp{\Forest'}=
\set{\comp{\firstVertex}{\Forest'}: \firstVertex \in \Vertices'}$.
For a fixed $\lev$ we observe how $\TreeLevell{\ttime}{\lev}$ changes from turn $\ttime-1$ to
$\ttime$.
Since the level function is monotonic, i.e., it satisfies $\level{\ttime-1}{v} \leqslant \level{\ttime}{v}$ (see
Observation~\ref{obs:monotonic}), the following hold:

\mythmspace{}

\begin{myobservation}
\label{obs:levelProps:subforest}
$\TreeLevell{\ttime-1}{\lev}$ is a subforest of $\TreeLevell{\ttime}{\lev}$. Also,
\MANUSCRIPTVERSION{
$$\VerticesOf{\TreeLevell{\ttime}{\lev}}= \VerticesOf{\TreeLevell{\ttime-1}{\lev}} \cup \set{v \in \VerticesOf{\Forest}:
\level{\ttime-1}{v} < \lev \leqslant \level{\ttime}{v}}.$$
}
\LNCSVERSION{
$\VerticesOf{\TreeLevell{\ttime}{\lev}}= \VerticesOf{\TreeLevell{\ttime-1}{\lev}} \cup \set{v \in \VerticesOf{\Forest}:
\level{\ttime-1}{v} < \lev \leqslant \level{\ttime}{v}}.$
}
\end{myobservation}

\mythmspace{}

We fix a turn $\ttime$ for which $\dist{\ttime}{\black_{\ttime}} < \infinity$, $\disp{\ttime}$ is defined, and $\level{\ttime}{\disp{\ttime}} \geqslant \constDist \level{\ttime-1}{\disp{\ttime}}$.
The idea is the following.
We let $\levlow=\level{\ttime-1}{\disp{\ttime}}$ and
$\levhi=\level{\ttime}{\disp{\ttime}}$.
For the purpose of the proof we need a function that describes some intermediate level between
$\level{\ttime-1}{}$ and $\level{\ttime}{}$.
We thus extend the level function to rational indices:
\MANUSCRIPTVERSION{
$$
\level{\ttime-1/2}{\firstVertex}= \left \lbrace
\begin{array}
{ll} \level{\ttime-1}{\firstVertex}& \text{if} \ \firstVertex= \disp{\ttime} \\
\level{\ttime}{\firstVertex}& \text{otherwise}.
\end{array}
\right.
$$
}
\LNCSVERSION{
$\level{\ttime-1/2}{\firstVertex}= \level{\ttime-1}{\firstVertex}$ if $\firstVertex= \disp{\ttime}$ and
$\level{\ttime-1/2}{\firstVertex}=\level{\ttime}{\firstVertex}$ otherwise.
}
Observe that $\level{\ttime}{\firstVertex}$ function is still monotonic in $\ttime$ after the
extension. We illustrate these definitions by example in Figure \ref{fig:fraclev}.
 \begin{figure}[htbp]\fontsize{10}{10}\selectfont
 
 \myfigspace{}

 \centering \def\svgscale{0.5} \resizebox{\textwidth}{!}{
\executeiffilenewer{fraclev.svg}{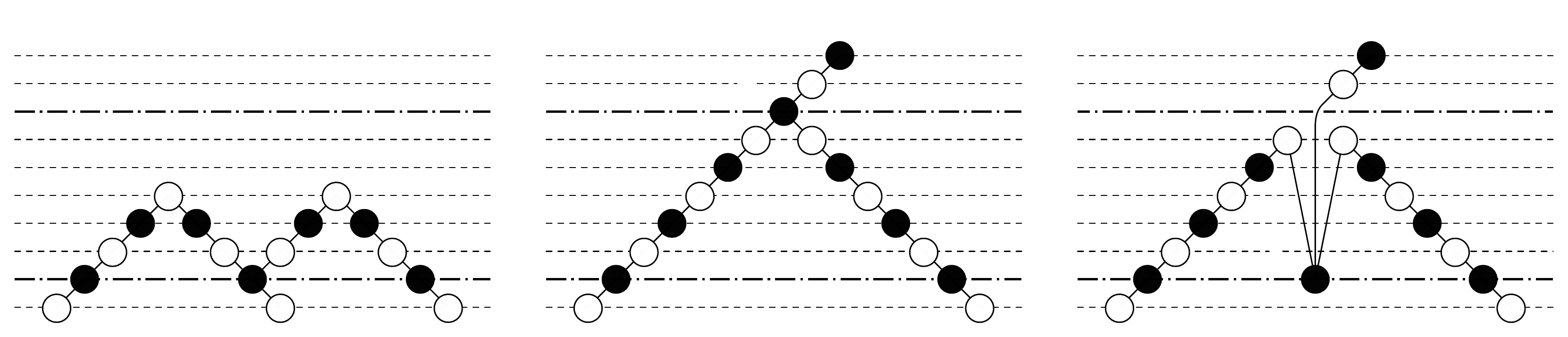} {inkscape -z -C --file=fraclev.svg --export-pdf=fraclev.pdf
--export-latex}
\begingroup
  \makeatletter
  \providecommand\color[2][]{
    \errmessage{(Inkscape) Color is used for the text in Inkscape, but the package 'color.sty' is not loaded}
    \renewcommand\color[2][]{}
  }
  \providecommand\transparent[1]{
    \errmessage{(Inkscape) Transparency is used (non-zero) for the text in Inkscape, but the package 'transparent.sty' is not loaded}
    \renewcommand\transparent[1]{}
  }
  \providecommand\rotatebox[2]{#2}
  \ifx\svgwidth\undefined
    \setlength{\unitlength}{956.59882813bp}
    \ifx\svgscale\undefined
      \relax
    \else
      \setlength{\unitlength}{\unitlength * \real{\svgscale}}
    \fi
  \else
    \setlength{\unitlength}{\svgwidth}
  \fi
  \global\let\svgwidth\undefined
  \global\let\svgscale\undefined
  \makeatother
  \begin{picture}(1,0.21493069)
    \put(0,0){\includegraphics[width=\unitlength]{fraclev.pdf}}
    \put(0.0352423,0.19376586){\color[rgb]{0,0,0}\makebox(0,0)[lb]{\smash{$\level{\ttime-1}{\bullet}$}}}
    \put(0.37422103,0.19376586){\color[rgb]{0,0,0}\makebox(0,0)[lb]{\smash{$\level{\ttime}{\bullet}$}}}
    \put(0.71319976,0.19376586){\color[rgb]{0,0,0}\makebox(0,0)[lb]{\smash{$\level{\ttime-1/2}{\bullet}$}}}
    \put(0.53211375,0.19376586){\color[rgb]{0,0,0}\makebox(0,0)[lb]{\smash{$b_t$}}}
    \put(0.47323853,0.14938865){\color[rgb]{0,0,0}\makebox(0,0)[lb]{\smash{$\check{b}_t$}}}
    \put(0.81043318,0.04234275){\color[rgb]{0,0,0}\makebox(0,0)[lb]{\smash{$\check{b}_t$}}}
    \put(0.87109253,0.19376586){\color[rgb]{0,0,0}\makebox(0,0)[lb]{\smash{$b_t$}}}
    \put(0.83184234,0.16522026){\color[rgb]{0,0,0}\makebox(0,0)[lb]{\smash{$w$}}}
    \put(0.49286362,0.16522026){\color[rgb]{0,0,0}\makebox(0,0)[lb]{\smash{$w$}}}
    \put(0.18956683,0.02249241){\color[rgb]{0,0,0}\makebox(0,0)[lb]{\smash{$w$}}}
    \put(0.32605039,0.03319699){\color[rgb]{0,0,0}\makebox(0,0)[lb]{\smash{$\levlow$}}}
    \put(0.32605039,0.1402429){\color[rgb]{0,0,0}\makebox(0,0)[lb]{\smash{$\levhi$}}}
    \put(0.66502911,0.03319699){\color[rgb]{0,0,0}\makebox(0,0)[lb]{\smash{$\levlow$}}}
    \put(0.66502911,0.1402429){\color[rgb]{0,0,0}\makebox(0,0)[lb]{\smash{$\levhi$}}}
  \end{picture}
\endgroup
  }
 \caption{Fractional levels}
 \label{fig:fraclev}

\myfigspace{}

 \end{figure}
We observe that on levels $\lev$ from $\levlow$ to $\levhi-1$ the
separate components of~$\TreeLevell{\ttime-1/2}{\lev}$ are merged in $\TreeLevell{\ttime}{\lev}$.
It is this merging that allows us to provide the charging scheme.
The level defined for fractional indices may be interpreted as an additional fractional turn
between $\ttime-1$ and $\ttime$.

Let us fix a level $\lev$.
In every turn $\ttime = 0, \frac{1}{2}, 1, 1\frac{1}{2}, \ldots,n$ a number of $\constTokens$
tokens is assigned to every connected component $\component$ in $\TreeLevell{\ttime}{\lev}$ such that
$\size{\component} \geqslant \constThreshold \lev$, where $\constTokens$ and
$\constThreshold$ are constants that we compute later.
Smaller components are not assigned any tokens.
Note that if $\lev$ is large, only large components are assigned tokens.
We plan to use these tokens to pay for the mini-max paths in each turn.
First, however, we describe how we maintain such an assignment on level $\lev$.

First we consider moving from turn $\ttime-1$ to turn $\ttime - 1/2$.
The forest $\TreeLevell{\ttime-1/2}{\lev}$ is obtained by adding the set of vertices
$
\new{\ttime-1/2} = \VerticesOf{\TreeLevell{\ttime-1/2}{\lev}} \setminus
\VerticesOf{\TreeLevell{\ttime-1}{\lev}} = \set{v \in \Black_{\ttime}\cup \White:
\level{\ttime-1}{v} < \lev \leqslant \level{\ttime - 1/2}{v}}
$
to $\TreeLevell{\ttime-1}{\lev}$ (see Observation~\ref{obs:levelProps:subforest}).
We want to add some structure to this process.
We divide transformation from $\TreeLevell{\ttime-1}{\lev}$ to
$\TreeLevell{\ttime-1/2}{\lev}$ into two sub-phases.
In the first sub-phase, the vertices of~$\new{\ttime-1/2}$ form new singleton components:
$\CComp' = \bigcup_{v \in \new{\ttime-1/2}} \set{ \tuple{ \set{v},\emptyset } }$, where
$\tuple{ \set{v},\emptyset }$ is a graph with only one vertex $v$ and without edges.
Together with the set of connected components of~$\TreeLevell{\ttime-1}{\lev}$ (referred to as
$\Comp{\TreeLevell{\ttime-1}{\lev}}$) they form a family
$\mathcal{I}=\Comp{\TreeLevell{\ttime-1}{\lev}} \cup \CComp'$.
In the second sub-phase, components in $\mathcal{I}$ merge whenever there is an edge of~$\Forest$
connecting them, finally becoming the connected components of $\TreeLevell{\ttime-1/2}{\lev}$.
Every component $\component \in \Comp{\TreeLevell{\ttime-1/2}{\lev}}$ can be assigned a set
of components $\mathcal{I}_{\component}=\set{\component_1 \ldots \component_k} \subseteq
\mathcal{I}$ that merged into $\component$.
There are two possible options:
\begin{enumerate}
[label=(\alph*)]
\item \label{case:there-is-big} there is the component $\component_i \in \mathcal{I}$ with size
$\size{\component_i} \geqslant \constThreshold \lev$, so $\component_i$ is already assigned
$\constTokens$ tokens
\item \label{case:all-small} every $\component_i \in \mathcal{I}$ satisfies $\size{\component_i}
< \constThreshold \lev$, so none of them is assigned any tokens.
\end{enumerate}
In case~\ref{case:there-is-big}, $\constTokens$ assigned to $\component_i$, which ceases to
exist, is now transferred to $\component$.
In case~\ref{case:all-small}, if $\size{\component} \geqslant \constThreshold \lev$, every
vertex $\firstVertex \in \component$ chips in with a
payment~of~$\frac{\constTokens}{\constThreshold \lev}$, so the vertices of~$\component$ pay in
total at least $\constTokens$.
We count the total amount that is paid at the end of the proof.

We now consider the transition from turn $\ttime-1/2$ to turn $\ttime$.
There is only one vertex, mainly $\disp{\ttime}$, which changes its level.
Its level increases from $\levlow$ to $\levhi$.
Level $\lev$ is only affected by this transition if $\levlow < \lev \leqslant \levhi$.
So, the forest $\TreeLevell{\ttime}{\lev}$ is formed from $\TreeLevell{\ttime - 1/2}{\lev}$
by adding $\disp{\ttime}$.
The only difference between $\Comp{\TreeLevell{\ttime-1/2}{\lev}}$ and
$\Comp{\TreeLevell{\ttime}{\lev}}$ is that some family of separate components of
$\Comp{\TreeLevell{\ttime-1/2}{\lev}}$ becomes connected by $\disp{\ttime}$ and forms a new connected component $\comp{\disp{\ttime}}{\TreeLevell{\ttime}{\lev}}$.
The set of components that merge into $\comp{\disp{\ttime}}{\TreeLevell{\ttime}{\lev}}$ is
precisely $\mathcal{I'}=\Comp{\comp{\disp{\ttime}} {\TreeLevell{\ttime}{\lev}} \setminus
\disp{\ttime}} \cup \set{ \tuple{ \disp{\ttime}, \emptyset } }$.
The way of assigning $\constTokens$ to $\comp{\disp{\ttime}}{\TreeLevell{\ttime}{\lev}}$ if
$\size{\comp{\disp{\ttime}}{\TreeLevell{\ttime}{\lev}}} \geqslant \constThreshold \lev$
is the same as in the transition from $\ttime-1$ to $\ttime-1/2$.
The difference is that now we want to utilize some of the assigned tokens to pay for the mini-max
path in turn $\ttime$.
Thus, we distinguish three cases now:
\begin{enumerate}
[label=(\roman*)]
\item \label{case:exactly-one} exactly one of the components $\component' \in \mathcal{I'}$
satisfies $\size{\component'} \geqslant \constThreshold \lev$, so $\component'$ is
assigned $\constTokens$ tokens,
\item \label{case:every-smaller} every $\component \in \mathcal{I'}$ satisfies
$\size{\component} < \constThreshold \lev$ so none of them is assigned tokens,
\item \label{case:two-big} two or more components $\component',\component'' \in
\mathcal{I'}$ satisfy $\size{\component'} \geqslant \constThreshold \lev$ and
$\size{\component''} \geqslant \constThreshold \lev$, so $\component'$ and $\component''$
are both already assigned $\constTokens$ tokens.
\end{enumerate}
Cases~\ref{case:exactly-one} and~\ref{case:every-smaller} are handled in the exactly same manner
as in the transition from $\ttime-1$ to $\ttime-1/2$.
The difference is that in case~\ref{case:two-big} we utilize $\constTokens$ tokens assigned to
$\component'$ while $\constTokens$ tokens assigned to $\component''$ transfer to
$\comp{\disp{\ttime}}{\TreeLevell{\ttime}{\lev}}$.

It remains to prove that the tokens utilized in turn $\ttime$ suffice to pay for
$\length{\pathhs{\ttime}{\black_{\ttime}}}$.
Observation~\ref{obs:distPathDisp} shows that $\level{\ttime}{\disp{\ttime}} =
\length{\pathhs{\ttime}{\black_{\ttime}}}$.
So we need to pay $\level{\ttime}{\disp{\ttime}}$ tokens when moving from turn $\ttime-1/2$ to
$\ttime$.
Let $\constThreshold:=( \constDist-1)/(\constDist+1)$.
By Claim~\ref{clm:size-of-path}, proved later on, case~\ref{case:two-big} occurs on at least $\bra{\levhi-\levlow}/2$ levels.
Since $\levlow \leqslant \levhi / \constDist$, we utilize at least $\constTokens
\bra{\levhi-\levhi / \constDist}/2= \frac{\constTokens
\bra{1-1/\constDist}}{2}\level{\ttime}{\disp{\ttime}}$ tokens.
Setting
$\constTokens:=2/(1-1/\constDist)$
allows paying the desired amount.

Now we count the sum of lengths of $\pathhs{\ttime}{\black_{\ttime}}$.
Every vertex pays $\frac{\constTokens}{\constThreshold l}$ at most once per level
and the highest level is not greater than $2n$, so the total
amount paid by a vertex over all the turns is bounded by
$\frac{\constTokens}{\constThreshold}\sum_{\lev=1}^{2n}
\frac{1}{\lev}\leqslant\frac{\constTokens}{\constThreshold} \bra{\ln \bra{2n} + 1}$.
Hence, the total amount paid by all vertices is at most $\frac{\constTokens}{\constThreshold} n
(\ln n + 1.7)$.
If we plug in the constants $\constThreshold = (\constDist-1)/(\constDist+1)$ and $\constTokens=2/(1-1/\constDist)$ into above bound, we obtain that
$
\sum_{\ttime \in \mathcal{T}}\length{\pathhs{\ttime}{\black_{\ttime}}} \leqslant
\frac{\constDist\bra{\constDist+1}}{\bra{\constDist -1} ^2} n (2 \ln n +3.4).
$\LNCSVERSION{\qed}\end{proof}

\mythmspace{}

To complete the proof of Lemma~\ref{lem:hard case} we move on to proving the following Claim.

\mythmspace{}

\begin{myclaim}
\label{clm:size-of-path} For a fixed $\ttime \in \range{n}$ let
$\lev_0 = \level{\ttime-1}{\disp{\ttime}}+1$ and
$\lev_1 = \floor{\bra{\level{\ttime-1}{\disp{\ttime}} +
\level{\ttime}{\disp{\ttime}}}/2}$.
For $\constThreshold=\frac{ \constDist-1}{\constDist+1}$ and $\lev \in \set{\lev_0,
\ldots,\lev_1}$ there exist two different vertices $\white_1,\white_2 \in \neighbours{\ttime}{\disp{\ttime}} \cap \Alive_{\ttime}$ such that $\comp{\white_1}{\TreeLevell{\ttime-1/2}{\lev}}$ and $\comp{\white_2}{\TreeLevell{\ttime-1/2}{\lev}}$ are two separate components of $\TreeLevell{\ttime-1/2}{\lev}$ and
$
\size{\comp{\white_i}{\TreeLevell{\ttime-1/2}{\lev}}} \geqslant \constThreshold\lev \text{ for } \ i\in \set{1,2}.
$
\end{myclaim}

\mythmspace{}

\mythmspace{}

\begin{proof} Fix $\lev \in \set{\lev_0,
\ldots,\lev_1}$.
By definition $\neighbours{\ttime}{\disp{\ttime}} \cap \Alive_{\ttime} \geq 2$. Thus, $\dist{\ttime}{\disp{\ttime}} < \infinity$ and $\secondDist{\ttime}{\disp{\ttime}} < \infinity$. We show that $\white_1 = \dir{\ttime}{\disp{\ttime}}$ and $\white_2=\secondDir{\ttime}{\disp{\ttime}}$ satisfy the desired conditions.

First note that $\white_1,\white_2 \in \VerticesOf{\TreeLevell{\ttime-1/2}{\lev}}$, because $\level{\ttime}{\white_i} \geq \level{\ttime}{\disp{\ttime}}-1 \geq \lev_1$ for $i \in \{1,2 \}$ due to Lemma~\ref{lem:plusMinusOne}. Note also that $\comp{\white_1}{\TreeLevell{\ttime-1/2}{\lev}}$ and $\comp{\white_2}{\TreeLevell{\ttime-1/2}{\lev}}$ are two separate components of $\TreeLevell{\ttime-1/2}{\lev}$ because the only path connecting $\white_1$ and $\white_2$ in $\Forest$ is through $\disp{\ttime}$ and $\disp{\ttime} \notin V(\TreeLevell{\ttime}{\lev})$ because $\level{\ttime-1/2}{\disp{\ttime}}=\level{\ttime-1}{\disp{\ttime}}< \lev_0 \leq \lev$. 

Due to Lemma~\ref{lem:plusMinusOne} the levels of vertices drop by at most one along $\pathh{\ttime}{\disp{\ttime}}$ and $\secondPath{\ttime}{\disp{\ttime}}$. Let $\pathOne_1$ be the prefix of $\pathh{\ttime}{\disp{\ttime}}$ of length $\levhi-\lev$ and $\pathOne_2$ be the prefix of $\secondPath{\ttime}{\disp{\ttime}}$ of length $\levhi-\lev$, where $\levhi=\level{\ttime}{\disp{\ttime}}$. It holds that if $v \in V(\pathOne_i)$ then $\level{\ttime-1/2}{v}=\level{\ttime}{v} \geq \lev$.
Because
$
\lev_0-1=
\level{\ttime-1}{\disp{\ttime}}\leqslant
\level{\ttime}{\disp{\ttime}}/\constDist=
\levhi/\constDist
$
and $\lev \leqslant \lev_1 = \floor{\bra{\lev_0-1+\levhi}/2}$ we have
$
\lev \leqslant \bra{\levhi/ \constDist + \levhi}/2= \bra{1/ \constDist+1} \cdot \levhi /2.
$
This implies
$
\size{\comp{\white}{\TreeLevell{\ttime}{\lev}}} \geqslant \levhi-\lev= \frac{2
\constDist}{\constDist+1} \cdot \frac{1+1/ \constDist}{2} \cdot \levhi-\lev \geqslant \frac{2
\constDist}{\constDist+1} \cdot \lev-\lev= \frac{ \constDist-1}{\constDist+1} \cdot \lev.
$
Setting $\constThreshold=(\constDist-1)/(\constDist+1)$ completes the proof of the claim.
\LNCSVERSION{\qed}\end{proof}

We can now put all the pieces together to prove our main result.

\begin{mytheorem}\label{thm:final}
$\sum_{\ttime \in \range{n}:\dist{\ttime}{\black_{\ttime}}<\infinity} \dist{\ttime}{\black_{\ttime}} \in \bigo{n \log n}.$ 
\end{mytheorem}

\mythmspace{}

\begin{proof}
By definition~\ref{def:presuf} we have $\pathh{\ttime}{\black_{\ttime}}=\pathhp{\ttime}{\black_{\ttime}}\cdot\pathhs{\ttime}{\black_{\ttime}}$. By Observation~\ref{obs:dieOnlyOnce} it holds that $\sum_{\ttime \in \range{n}:\dist{\ttime}{\black_{\ttime}}<\infinity} \length{\pathhp{\ttime}{\black_{\ttime}}} \leq 2n$. If $\disp{\ttime}$ is undefined, then $\pathhs{\ttime}{\black_{\ttime}}$ is empty so its length is $0$. For the cases when $\disp{\ttime}$ is defined and $\level{\ttime}{\disp{\ttime}} < \constDist \level{\ttime-1}{\disp{\ttime}}$ we have $\sum_{\ttime \in \range{n}:\dist{\ttime}{\black_{\ttime}}<\infinity} \length{\pathhs{\ttime}{\black_{\ttime}}} \leq 2 \constDist n+\constDist n \log n$. For the cases when $\disp{\ttime}$ is defined and $\level{\ttime}{\disp{\ttime}} \geq \constDist \level{\ttime-1}{\disp{\ttime}}$ we have $\sum_{\ttime \in \range{n}:\dist{\ttime}{\black_{\ttime}}<\infinity} \length{\pathhs{\ttime}{\black_{\ttime}}} \leq \frac{\constDist(\constDist+1)}{(\constDist-1)^2} n (2 \ln n +3.4) +n$. This gives the theorem statement for any $\constDist > 1$.   
\LNCSVERSION{\qed}\end{proof}

\mythmspace{}
 \bibliographystyle{alpha}

\clearpage

\appendix

\section{The mini-max game}\label{app:minimax}

\obsmonotonic*

\begin{proof}
We only prove monotonicity of $\dist{\ttime}{v}$ here, the proof of monotonicity of $\secondDist{\ttime}{v}$ is analogous.
Fix $\ttime \in \range{n} $. We claim that $\dist{\ttime-1}{v} \leq \dist{\ttime}{v}$ provided that $v \in V(\Forest_{\ttime-1})$.
Let $\RootedTreeOne$ be a component of $v$ in $\Forest_{\ttime-1}$ rooted at $v$ and $\RootedTreeTwo$ be a component of $v$ in $\Forest_{\ttime}$ rooted at $v$. Note that $\RootedTreeTwo$ is a subtree of $\RootedTreeOne$.
By definition it holds that $\dist{\ttime-1}{v}=\treeDist{\RootedTreeOne}{v}$ and $\dist{\ttime}{v}=\treeDist{\RootedTreeTwo}{v}$.
We prove our claim by a bottom-up induction on $\RootedTreeOne$. The inductive hypothesis we prove is that $\treeDist{\RootedTreeOne}{u} \leq \treeDist{\RootedTreeTwo}{u}$ for $u \in V(\RootedTreeOne)$. 

First assume that $u$ is black. The neighborhood of $u$ does not change from $t-1$ to $t$, since after $u$ is presented, it never changes its neighborhood. Therefore $\children{\RootedTreeOne}{u} = \children{\RootedTreeTwo}{u}$. If $u$ is a leaf of $\RootedTreeOne$, then  $\treeDist{\RootedTreeOne}{u} =
\infinity = \treeDist{\RootedTreeTwo}{u}$.
If $u$ has children then by induction hypothesis on children it follows that
$$
\treeDist{\RootedTreeOne}{u} = \min_{\white \in \children{\RootedTreeOne}{u}} \treeDist{\RootedTreeOne}{\white}+1 \leqslant^{\text{ind}} \min_{\white \in \children{\RootedTreeTwo}{u}}
\treeDist{\RootedTreeTwo}{\white}+1
=\treeDist{\RootedTreeTwo}{u}.
$$

Now assume that $u$ is white. The neighborhood of $u$ can change from $t-1$ to $t$, but it can only increase, meaning that $\children{\RootedTreeOne}{u} \subseteq \children{\RootedTreeTwo}{u}$.
If $u$ is a leaf in $\RootedTreeOne$, then by the fact that
$\treeDist{\RootedTreeTwo}{u}$ is non-negative it holds that
$\treeDist{\RootedTreeOne}{u} = 0 \leqslant
\treeDist{\RootedTreeTwo}{u}$.
If $u$ has children in $\RootedTreeOne$, then by induction hypothesis we infer that
\begin{multline*}
\treeDist{\RootedTreeOne}{u} = 
\max_{\black \in \children{\RootedTreeOne}{u}} \treeDist{\RootedTreeOne}{\black}+1\leqslant^{\text{ind}}
\max_{\black \in \children{\RootedTreeOne}{u}} \treeDist{\RootedTreeTwo}{\black}+1
\leqslant\\
\leqslant \max_{\black \in \children{\RootedTreeTwo}{u}} \treeDist{\RootedTreeTwo}{\black}+1=\treeDist{\RootedTreeTwo}{u}.
\end{multline*}
\end{proof}

\lempathineq*

\begin{proof}
Fix $\ttime_0\in\range{n} $ and $\ttime \in \set{\ttime_0, \ldots,n }$.
Let $\RootedTreeOne$ be the connected component of $\black_{\ttime_0}$ in $\Forest_t$ rooted at $\black_{\ttime_0}$. 
We prove our lemma by a bottom-up induction on $\RootedTreeOne$.

Before we state inductive hypothesis, we need to introduce a few more definitions.
We denote as $\pathTwo_{\RootedTreeOne}\bra{\firstVertex}$ the shortest alternating path from $\firstVertex$ to a free white vertex in the subtree of $\RootedTreeOne$ rooted at $\firstVertex$.
Here, an augmenting path is taken with respect to $M$ and a white vertex is considered free if it is free in the subtree.
If there is no such alternating path we write $\pathTwo_{\RootedTreeOne}\bra{\firstVertex}=\bot$ for which length $\length{\bot}=\infinity$.
It is not hard to see that such a path is either a single free white vertex, or it starts with an unmatched edge if $\firstVertex$ is black and with a matched edge if $\firstVertex$ is white.
When we say that a vertex is unmatched from below, we mean that all its incident edges leading to its children in $\RootedTreeOne$ are free from matching.
Similarly, we say that a vertex is unmatched from above, if the edge leading to its parent in $\RootedTreeOne$ is free from matching.

Our inductive hypothesis states the following. If $\firstVertex$ is 
\begin{enumerate}
\item[(a)] black and unmatched from below or
\item[(b)] white and unmatched from above
\end{enumerate}
then $\length{\pathTwo_{\RootedTreeOne}\bra{\firstVertex}} \leqslant \treeDist{\RootedTreeOne}{\firstVertex}$.

Let us first assume that $\firstVertex$ is black.
If $\firstVertex$ is a leaf in $\RootedTreeOne$, then $\pathTwo_{\RootedTreeOne}\bra{\firstVertex}=\bot$, so $\length{\pathTwo_{\RootedTreeOne}\bra{\firstVertex}} = \infinity = \treeDist{\RootedTreeOne}{\firstVertex}$, so the hypothesis is fulfilled. Let us consider the case when $\firstVertex$ is not a leaf. We consider two sub-cases of this case.
First assume that $\pathTwo_{\RootedTreeOne}\bra{\firstVertex}=\bot$. This implies that $\pathTwo_{\RootedTreeOne}(\white)=\bot$ for all $\white \in \children{\RootedTreeOne}{\firstVertex}$.
By induction it holds that
$$
\treeDist{\RootedTreeOne}{\firstVertex} = \min_{\white \in \children{\RootedTreeOne}{\firstVertex}} \treeDist{\RootedTreeOne}{\white} +1 =^{\text{ind}} \infinity+1=\infinity.
$$ 
Now assume that $\pathTwo_{\RootedTreeOne}\bra{\firstVertex} \neq \bot$.
Let $\white$ be a successor of $\firstVertex$ on $\pathTwo_{\RootedTreeOne}\bra{\firstVertex}$ and let $\white'=\treeDir{\RootedTreeOne}{\firstVertex}$.
Then by the power of induction we infer that
$$
\treeDist{\RootedTreeOne}{\firstVertex} = \treeDist{\RootedTreeOne}{\white'}+1 \geqslant^{\text{ind}} \length{\pathTwo_{\RootedTreeOne}\bra{\white'}}+1 \geqslant \length{\pathTwo_{\RootedTreeOne}\bra{\white}}+1=
\length{\pathTwo_{\RootedTreeOne}\bra{\firstVertex}}.
$$ 

Let us now assume that $\firstVertex$ is white.
If $\firstVertex$ is a leaf, then $\length{\pathTwo_{\RootedTreeOne}\bra{\firstVertex}}=0=\treeDist{\RootedTreeOne}{\firstVertex}$ so the hypothesis is fulfilled.
If $\firstVertex$ is not a leaf, we again consider two cases.
First assume that $\pathTwo_{\RootedTreeOne}\bra{\firstVertex}=\bot$.
In such case there has to be a vertex $\black' \in \children{\RootedTreeOne}{\firstVertex}$ matched to $\firstVertex$, otherwise $\firstVertex$ is free and the path $\pathTwo_{\RootedTreeOne}\bra{\firstVertex}$ exists.
In such case $\pathTwo_{\RootedTreeOne}\bra{\black'}=\bot$ otherwise there would exist augmenting path also from $\firstVertex$.
Then by inductive assumption $\treeDist{\RootedTreeOne}{\black'} = \infinity$.
From that it follows that
$$
\treeDist{\RootedTreeOne}{\firstVertex} = \max_{\black \in \children{\RootedTreeOne}{\firstVertex}} \treeDist{\RootedTreeOne}{\black}+1 \geqslant\treeDist{\RootedTreeOne}{\black'}+1 =^{\text{ind}} \infinity+1= \infinity.
$$
Now assume that $\pathTwo_{\RootedTreeOne}\bra{\firstVertex} \neq \bot$.
In this case either $\firstVertex$ is free in its subtree and $\length{\pathTwo_{\RootedTreeOne}\bra{\firstVertex}} = 0 = \treeDist{\RootedTreeOne}{\firstVertex}$, or there is a vertex $\black' \in \children{\RootedTreeOne}{\firstVertex}$ matched to $\firstVertex$.
But then by the power of induction we infer that
$$
\treeDist{\RootedTreeOne}{\firstVertex} = \max_{\black \in \children{\RootedTreeOne}{\firstVertex}} \treeDist{\RootedTreeOne}{\black}  +1
\geqslant \treeDist{\RootedTreeOne}{\black'} +1
\geqslant^{\text{ind}} \length{\pathTwo_{\RootedTreeOne}\bra{\black'}}+1
=\length{\pathTwo_{\RootedTreeOne}\bra{\firstVertex}}.
$$

This proves the hypothesis, in particular this proves that $\length{\pathTwo_{\RootedTreeOne}\bra{\black_{\ttime_0}}} \leqslant \treeDist{\RootedTreeOne}{\black_{\ttime_0}}=\dist{\ttime}{\black_{\ttime_0}}$. Note that if $\dist{\ttime}{\black_{\ttime_0}} < \infinity$, then $\length{\pathTwo_{\RootedTreeOne}\bra{\black_{\ttime_0}}} < \infinity$, and that by definition implies that $\pathTwo_t = \pathTwo_{\RootedTreeOne}\bra{\black_{\ttime_0}}$ exists and $\length{\pathTwo_{\ttime}}=\length{\pathTwo_{\RootedTreeOne}\bra{\black_{\ttime_0}}} \leqslant \dist{\ttime}{\black_{\ttime_0}}$.
\end{proof}

\clearpage

\section{Dead vertices}\label{app:dead}

\lemdistBlinf*

\begin{proof}
\item
\paragraph{"$\Leftarrow$":} 
Assume that $\dist{\ttime}{\black_{\ttime_0}}<\infinity$. Then, by Lemma~\ref{lem:pathineq}, an augmenting path exists from $\black_{\ttime_0}$ for a maximum matching $M$ where $\black_{\ttime_0}$ is free. By Hall's theorem this implies that $\black_{\ttime_0}$ does not break \hc{}.
\paragraph{"$\Rightarrow$":}
Let $\RootedTreeOne$ be the connected component of $\black_{\ttime_0}$ in $\Forest_{\ttime}$ rooted at $\black_{\ttime_0}$. Assume $\dist{\ttime}{\black_{\ttime_0}}=\infinity$. We are going to define a subtree $\RootedTreeTwo$ of $\RootedTreeOne$ such that $V(\RootedTreeTwo) \cap \Black_{\ttime}$ is precisely the set $X$ that certifies that $\black_{\ttime_0}$ breaks \hc{} in time $\ttime$. 

For that we first define a forest $\RootedForest$ that is a subgraph of $\RootedTreeOne$ and pick one component of $\RootedForest$ to be $\RootedTreeTwo$. Let $\RootedForest=(\White_{\Forest} \cup \Black_{\Forest},E_{\Forest})$, where $\White_{\Forest}=\White \cap V(\RootedTreeOne)$ are the white vertices of $\RootedTreeOne$, $\Black_{\Forest}=\Black_{\ttime} \cap V(\RootedTreeOne)$ are the black vertices of $\RootedTreeOne$ and $E_{\Forest}=E_1 \cup E_2$, where
\begin{align*}
E_1&=\set{ (\parent{\RootedTreeOne}{\white},\white): \white \in \White_{\Forest} }, \\
E_2&=\set{ (\white,\treeDir{\RootedTreeOne}{\white}): \white \in \White_{\Forest} \text{ and } \treeDir{\RootedTreeOne}{\white} \text{ exists} }.
\end{align*}
Next we let $\RootedTreeTwo$ be the connected component of $\black_{\ttime_0}$ in $\RootedForest$ rooted at $\black_{\ttime_0}$. The construction of $\RootedTreeTwo$ is shown in Figure~\ref{fig:tree_e1_e2}.
\begin{figure}[htbp]\fontsize{10}{10}\selectfont
\centering \def\svgscale{0.9} \resizebox{\textwidth}{!}{
\executeiffilenewer{tree_e1_e2.svg}{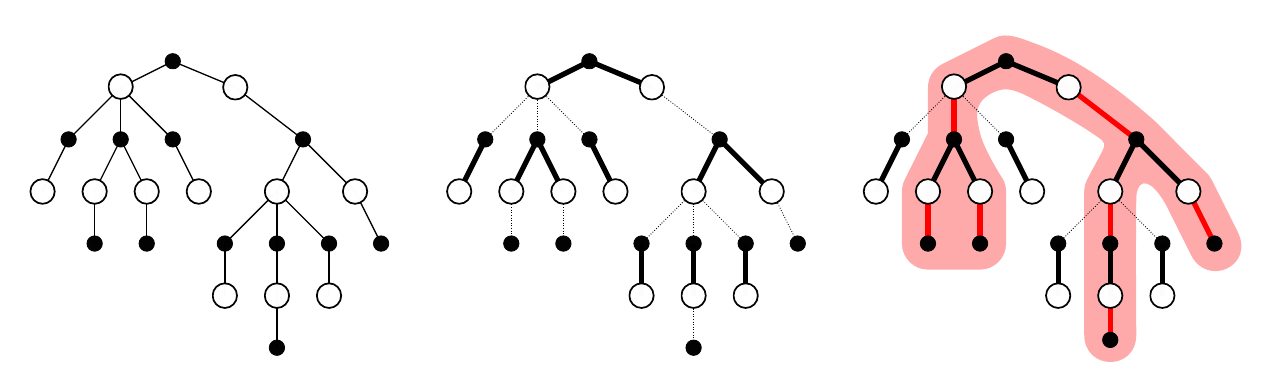} {inkscape -z -C --file=tree_e1_e2.svg --export-pdf=tree_e1_e2.pdf
--export-latex}
\begingroup
  \makeatletter
  \providecommand\color[2][]{
    \errmessage{(Inkscape) Color is used for the text in Inkscape, but the package 'color.sty' is not loaded}
    \renewcommand\color[2][]{}
  }
  \providecommand\transparent[1]{
    \errmessage{(Inkscape) Transparency is used (non-zero) for the text in Inkscape, but the package 'transparent.sty' is not loaded}
    \renewcommand\transparent[1]{}
  }
  \providecommand\rotatebox[2]{#2}
  \ifx\svgwidth\undefined
    \setlength{\unitlength}{366.03279174bp}
    \ifx\svgscale\undefined
      \relax
    \else
      \setlength{\unitlength}{\unitlength * \real{\svgscale}}
    \fi
  \else
    \setlength{\unitlength}{\svgwidth}
  \fi
  \global\let\svgwidth\undefined
  \global\let\svgscale\undefined
  \makeatother
  \begin{picture}(1,0.30806236)
    \put(0,0){\includegraphics[width=\unitlength,page=1]{tree_e1_e2.pdf}}
    \put(0.12767437,0.0302336){\color[rgb]{0,0,0}\makebox(0,0)[lb]{\smash{$\RootedTreeOne$}}}
    \put(0.45551404,0.0302336){\color[rgb]{0,0,0}\makebox(0,0)[lb]{\smash{$E_1$}}}
    \put(0.78335289,0.0302336){\color[rgb]{0,0,0}\makebox(0,0)[lb]{\smash{$E_2$}}}
    \put(0.1322848,0.27201539){\color[rgb]{0,0,0}\makebox(0,0)[lb]{\smash{$\black_{t_0}$}}}
  \end{picture}
\endgroup
  }
\caption{The construction of rooted tree $\RootedTreeTwo$ in the proof of Lemma~\ref{lem:distBlinf}.}
\label{fig:tree_e1_e2}
\end{figure}
Before we proceed, we prove a helpful observation that $\treeDist{\RootedTreeTwo}{u} =\infinity$ holds for all $u \in V(\RootedTreeTwo)$. To that end we assume for the sake of contradiction that there is $u \in V(\RootedTreeTwo)$ such that $\treeDist{\RootedTreeTwo}{u}<\infinity$ and we pick the closest to the root such vertex. Then $\treeDist{\RootedTreeTwo}{\parent{\RootedTreeTwo}{u}} = \infinity$. We consider two cases. If $u$ is white then $v=\parent{\RootedTreeTwo}{u}$ is black and thus
$$
\infinity = \treeDist{\RootedTreeTwo}{v} = \min_{\white \in \children{\RootedTreeTwo}{v}}\treeDist{\RootedTreeTwo}{\white}+1\leqslant \treeDist{\RootedTreeTwo}{u} + 1 <\infinity.
$$
This gives a contradiction. If, on the other hand, $u$ is black then $v=\parent{\RootedTreeTwo}{u}$ is white and hence, due to the way $\RootedTreeTwo$ is constructed, $v$ has only one child $u$ in $\RootedTreeTwo$. Thus
$$
 \infinity = \treeDist{\RootedTreeTwo}{v} = \max_{\black \in \children{\RootedTreeTwo}{v}}\treeDist{\RootedTreeTwo}{\black}+1= \treeDist{\RootedTreeTwo}{u} + 1 <\infinity.
$$
This also gives a contradiction.

Let us denote white vertices of $\RootedTreeTwo$ as $\White_{\TreeTwo}=V(\RootedTreeTwo) \cap \White$ and black vertices of $\RootedTreeTwo$ as $\Black_{\TreeTwo}=V(\RootedTreeTwo) \cap \Black_{\ttime}$.
Now let $X=\Black_{\TreeTwo}$. We prove that conditions (1)-(3) of Definition~\ref{def:breakHC} hold for $X$.
Condition (1) trivially holds.
Let us consider condition~(2). Since $\treeDist{\RootedTreeTwo}{\white} =\infinity$ for all $\white \in \White_{\TreeTwo}$, every $\white \in \White_{\TreeTwo}$ has a child in $\RootedTreeTwo$ and by the way $\RootedTreeTwo$ is constructed this child is unique. On the other hand each black vertex $\black \in \Black_{\TreeTwo} \setminus \{ \black_{\ttime_0} \}$ has a unique parent. This gives a one-to-one map $f:\Black_{\TreeTwo} \setminus \{ \black_{\ttime_0} \} \rightarrow \White_{\TreeTwo}$ (such that $f\bra{\black}:=\parent{\RootedTreeTwo}{\black}$) and implies $\abs{\White_{\TreeTwo}}=\abs{\Black_{\TreeTwo} \setminus \{ \black_{\ttime_0} \}}$. 

We move on to proving condition (3). Let $X' \varsubsetneq X$.
If $\black_{\ttime_0}\notin X'$ then $\abs{X'}=\abs{f\bra{X'}}\leqslant\abs{\neighbours{\ttime}{X'}}$.
So, we can assume that $\black_{\ttime_0}\in X'$.
Then there exists $\black' \in X \setminus X'$ such that $\parent{\RootedTreeTwo}{\parent{\RootedTreeTwo}{\black'}}\in X'$.
Thus
$$
\abs{X'}=\abs{X' \setminus \set{\black_{\ttime_0}}}+1=\abs{f\bra{X'\setminus\set{\black_{\ttime_0}}}}+1 = \abs{f\bra{X'\setminus\set{\black_{\ttime_0}}}\cup\set{\parent{\RootedTreeTwo}{\black'}}} \leqslant         \abs{\neighbours{\ttime}{X'}}.
$$
This proves the minimality of $X$ and ends the proof of the lemma.
\end{proof}

\lemmmPathsAlive*

\begin{proof}
Observe that
\begin{equation}
\label{eq:minimaxSecond}
\treeDist{\RootedTreeOne}{\black'}\leqslant\secondDist{\ttime}{\black'}\quad\text{for} \ \black'\in \Black_\ttime.
\end{equation}
Let $\mmpath{\RootedTreeOne}{\firstVertex}=\bra{\firstVertex^{\bra{1}},\firstVertex^{\bra{2}},\ldots,\firstVertex^{\bra{k}}}$.
We prove that $\firstVertex^{\bra{i}}$ is alive by induction on $i=1,2,\ldots,k$.
Vertex $\firstVertex^{\bra{1}}=\firstVertex$ is alive by the assumptions of the lemma.
We now assume that $\firstVertex^{\bra{i}}$ is alive and aim to prove that $\firstVertex^{\bra{i
+1}}$ is alive too.
If $\firstVertex^{\bra{i}}\in\White$ then by Observation~\ref{obs:aliveness} we know that all neighbours of $\firstVertex^{\bra{i}}$ are alive.
In particular $\firstVertex^{\bra{i+1}}\in\Alive_{\ttime}$.
We are left with the case when $\black:=\firstVertex^{\bra{i}}\in\Black_{\ttime}$ and  $\white:=\firstVertex^{\bra{i+1}}=\treeDir{\RootedTreeOne}{\black}\in\White$.
We aim to prove that $\white$ is alive, i.e. $\dist{\ttime}{\white}<\infinity$.
Let $\RootedTreeTwo$ be the connected component of $\firstVertex$ in $\Forest_{\ttime}$ rooted in $\white$. For the reference see Figure~\ref{pic:easy case 2}.
\begin{figure}[htbp]\fontsize{10}{10}\selectfont
\centering \def\svgscale{0.9} \resizebox{\textwidth}{!}{
\executeiffilenewer{easy_case_2.svg}{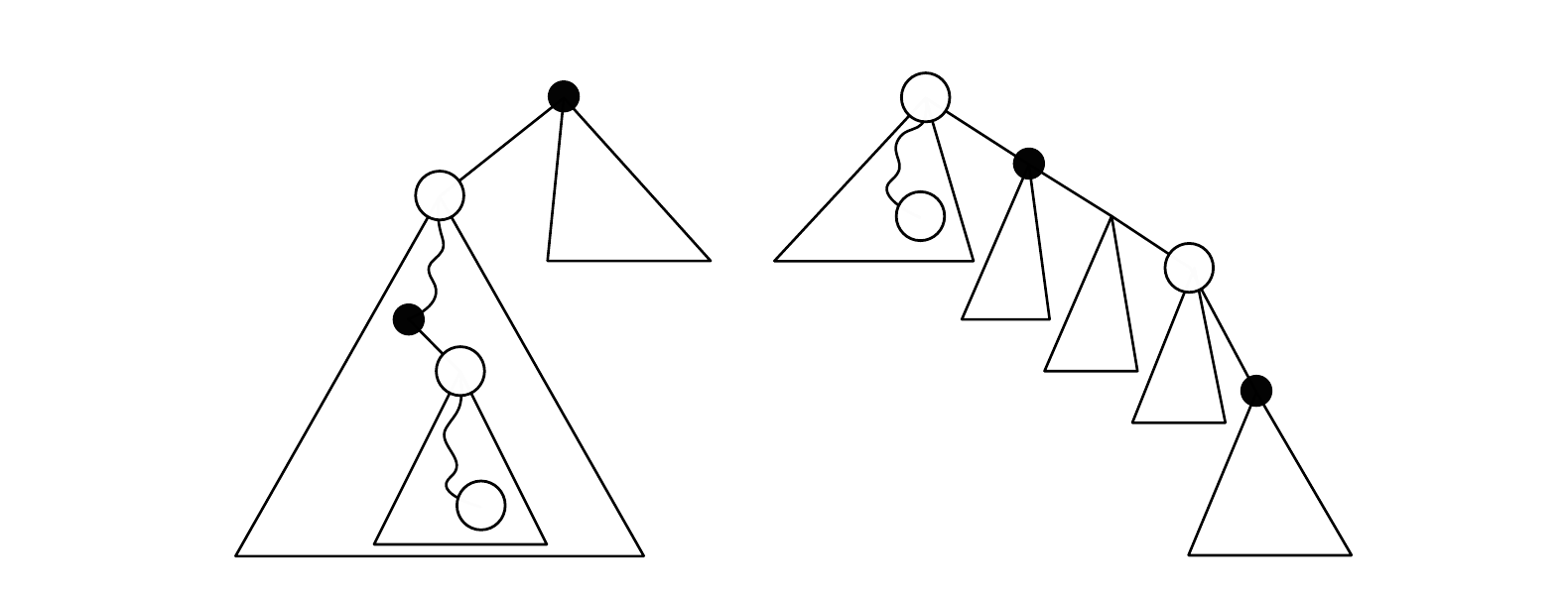} {inkscape -z -C --file=easy_case_2.svg --export-pdf=easy_case_2.pdf
--export-latex}
\begingroup
  \makeatletter
  \providecommand\color[2][]{
    \errmessage{(Inkscape) Color is used for the text in Inkscape, but the package 'color.sty' is not loaded}
    \renewcommand\color[2][]{}
  }
  \providecommand\transparent[1]{
    \errmessage{(Inkscape) Transparency is used (non-zero) for the text in Inkscape, but the package 'transparent.sty' is not loaded}
    \renewcommand\transparent[1]{}
  }
  \providecommand\rotatebox[2]{#2}
  \ifx\svgwidth\undefined
    \setlength{\unitlength}{454.97710941bp}
    \ifx\svgscale\undefined
      \relax
    \else
      \setlength{\unitlength}{\unitlength * \real{\svgscale}}
    \fi
  \else
    \setlength{\unitlength}{\svgwidth}
  \fi
  \global\let\svgwidth\undefined
  \global\let\svgscale\undefined
  \makeatother
  \begin{picture}(1,0.38077735)
    \put(0,0){\includegraphics[width=\unitlength,page=1]{easy_case_2.pdf}}
    \put(0.35372081,0.34146813){\color[rgb]{0,0,0}\makebox(0,0)[lb]{\smash{$r$}}}
    \put(0.13612657,0.15354627){\color[rgb]{0,0,0}\makebox(0,0)[lb]{\smash{$\RootedTreeOne$:}}}
    \put(0.37350128,0.23267149){\color[rgb]{0,0,0}\makebox(0,0)[lb]{\smash{$B$}}}
    \put(0.25151734,0.2738828){\color[rgb]{0,0,0}\makebox(0,0)[lb]{\smash{$v$}}}
    \put(0.31745478,0.13706262){\color[rgb]{0,0,0}\makebox(0,0)[lb]{\smash{$w$}}}
    \put(0.25151686,0.04145233){\color[rgb]{0,0,0}\makebox(0,0)[lb]{\smash{$A$}}}
    \put(0.42625143,0.0546401){\color[rgb]{0,0,0}\makebox(0,0)[lb]{\smash{$\mmpath{\RootedTreeOne}{v}$}}}
    \put(0.58450139,0.34146813){\color[rgb]{0,0,0}\makebox(0,0)[lb]{\smash{$w$}}}
    \put(0.52515712,0.22607737){\color[rgb]{0,0,0}\makebox(0,0)[lb]{\smash{$A$}}}
    \put(0.4790011,0.15354627){\color[rgb]{0,0,0}\makebox(0,0)[lb]{\smash{$\RootedTreeTwo$:}}}
    \put(0.81198639,0.14695262){\color[rgb]{0,0,0}\makebox(0,0)[lb]{\smash{$r$}}}
    \put(0.79879776,0.04474963){\color[rgb]{0,0,0}\makebox(0,0)[lb]{\smash{$B$}}}
    \put(0.77901624,0.20629689){\color[rgb]{0,0,0}\makebox(0,0)[lb]{\smash{$v$}}}
    \put(0.23833004,0.1436558){\color[rgb]{0,0,0}\makebox(0,0)[lb]{\smash{$b$}}}
    \put(0.67351709,0.27882751){\color[rgb]{0,0,0}\makebox(0,0)[lb]{\smash{$b$}}}
    \put(0.46911015,0.0117809){\color[rgb]{0,0,0}\makebox(0,0)[lb]{\smash{$\Forest_{\ttime-1}$}}}
    \put(0,0){\includegraphics[width=\unitlength,page=2]{easy_case_2.pdf}}
  \end{picture}
\endgroup
  }
\caption{Aliveness of $\mmpath{\ttime}{\firstVertex}$}
\label{pic:easy case 2}
\end{figure}

Then
$$
\dist{\ttime}{\white}=\treeDist{\RootedTreeTwo}{\white}=\max_{\black'\in\children{\RootedTreeTwo}{\white}}\treeDist{\RootedTreeTwo}{\black'}.
$$
So it suffices to prove that $\treeDist{\RootedTreeTwo}{\black'}<\infinity$ for any child $\black'$ of $\white$ in $\RootedTreeTwo$. 
If $\black'=\black$ then
$$
\treeDist{\RootedTreeTwo}{\black}\leqslant^{\text{by(\ref{eq:minimaxSecond})}}
\secondDist{\ttime}{\black}<
\infinity
$$
where the last inequality holds by aliveness of $\black$.
If $\black'\neq\black$ then the subtree of $\black'$ is identical for $\RootedTreeOne$ and $\RootedTreeTwo$ and as a consequence
\begin{multline*}
\treeDist{\RootedTreeTwo}{\black'}=
\treeDist{\RootedTreeOne}{\black'}\leqslant
\max_{\black''\in\children{\RootedTreeOne}{\white}}\treeDist{\RootedTreeOne}{\black''}=
\\=
\treeDist{\RootedTreeOne}{\white}-1=
\treeDist{\RootedTreeOne}{\black}-2\leqslant^{\text{by(\ref{eq:minimaxSecond})}}
\secondDist{\ttime}{\black}-2<
\infinity.
\end{multline*}
\end{proof}

\lemaliveneighbour*

\begin{proof} Let $\RootedTreeOne$ be the component of $\black_{\ttime}$ in $\Forest_{\ttime}$ rooted at $\black_{\ttime}$.
By Lemma~\ref{lem:distBlinf} we have that
$$
\infinity> \dist{\ttime}{\black_{\ttime}} = \min_{\white \in
\children{\RootedTreeOne}{\black_{\ttime}}} \treeDist{\RootedTreeOne}{\white} + 1 =
\min_{\white \in
\children{\RootedTreeOne}{\black_{\ttime}}} \dist{\ttime-1}{\white} + 1.
$$
Hence, the new vertex $\black_{\ttime}$ has a white neighbour $\white'$ such that
$\dist{\ttime-1}{\white'} <\infinity$.
By definition $\white'$ is alive in turn $\ttime-1$.
\end{proof}

\lemifatleasttwoalive*

\begin{proof}
We first prove that no vertex dies in turn $\ttime$.
Let $u \in \Black_{\ttime-1} \cup \White$ for some $\ttime \in \range{n} $ and assume that $\black_{\ttime}$ is as stated in the lemma. 
Let $\RootedTreeOne$ be the connected component of $u$ in $\Forest_{\ttime-1}$ rooted at $u$ and let $\RootedTreeTwo$ be the component of $u$ in $\Forest_{\ttime}$ rooted at $u$. If $u$ is not connected to $\black_{\ttime}$ in $\Forest_{\ttime}$, then by definition $\secondDist{\ttime-1}{u}=\secondDist{\ttime}{u}$ since $\RootedTreeOne=\RootedTreeTwo$. Hence, $u$ does not die in turn $\ttime$.

Otherwise $\neighbours{\ttime}{\black_{\ttime}}$ contains a vertex of $\RootedTreeOne$, let it be $\white$. 
For this case, to prove that $u$ does not die in turn $\ttime$, we show the following:
\begin{align*}
\dist{\ttime-1}{u} < \infinity &\Rightarrow \dist{\ttime}{u} < \infinity \text{ and}\\
\secondDist{\ttime-1}{u} < \infinity &\Rightarrow \secondDist{\ttime}{u} < \infinity.
\end{align*}

Let $\Forest_{\ttime}- \set{\white, \black_{\ttime} }$ be the graph obtained by removing edge $\set{\white, \black_{\ttime}}$ from $\Forest_{\ttime}$. Let $\RootedTreeThree$ be the connected component of $\black_{\ttime}$ in $\Forest_{\ttime} - \set{\white, \black_{\ttime} }$, rooted at $\black_{\ttime}$. The only difference between $\RootedTreeOne$ and $\RootedTreeTwo$ is that $\black_{\ttime}$ becomes a child of $\white$ in $\RootedTreeTwo$ and $\RootedTreeThree$ is attached to $\RootedTreeOne$ via edge $(\white,\black_{\ttime})$. This is illustrated in Figure~\ref{fig:ststr12} to the left.
Note that the only vertices who change their mini-max revenue are on the path $\pathOne_{u \white}$ from $u$ to $\white$ in $\RootedTreeOne$. Also, in the proof of Observation~\ref{obs:monotonic} we proved that $\treeDist{\RootedTreeOne}{v} \leq \treeDist{\RootedTreeTwo}{v}$ for $v \in V(\RootedTreeOne)$.
These two observations imply that it suffices to show that $\treeDist{\RootedTreeOne}{\white} < \infinity \Rightarrow \treeDist{\RootedTreeTwo}{\white} < \infinity$. 
Due to the assumption of the lemma, there exists $\white' \in \children{\RootedTreeTwo}{\black_{\ttime}}$ such that $\treeDist{\RootedTreeTwo}{\white'}=\dist{\ttime-1}{\white'} < \infinity$. Thus
$$
\treeDist{\RootedTreeTwo}{\black_{\ttime}}=\min_{\white'' \in \children{\RootedTreeTwo}{\black_{\ttime}}}\treeDist{\RootedTreeTwo}{\white''} + 1 \leqslant \treeDist{\RootedTreeTwo}{\white'}+1< \infinity.
$$
For the reference see Figure~\ref{fig:ststr12} to the right.
\begin{figure}[htbp]\fontsize{10}{10}\selectfont
\centering \def\svgscale{0.9} \resizebox{\textwidth}{!}{
\executeiffilenewer{ststr12.svg}{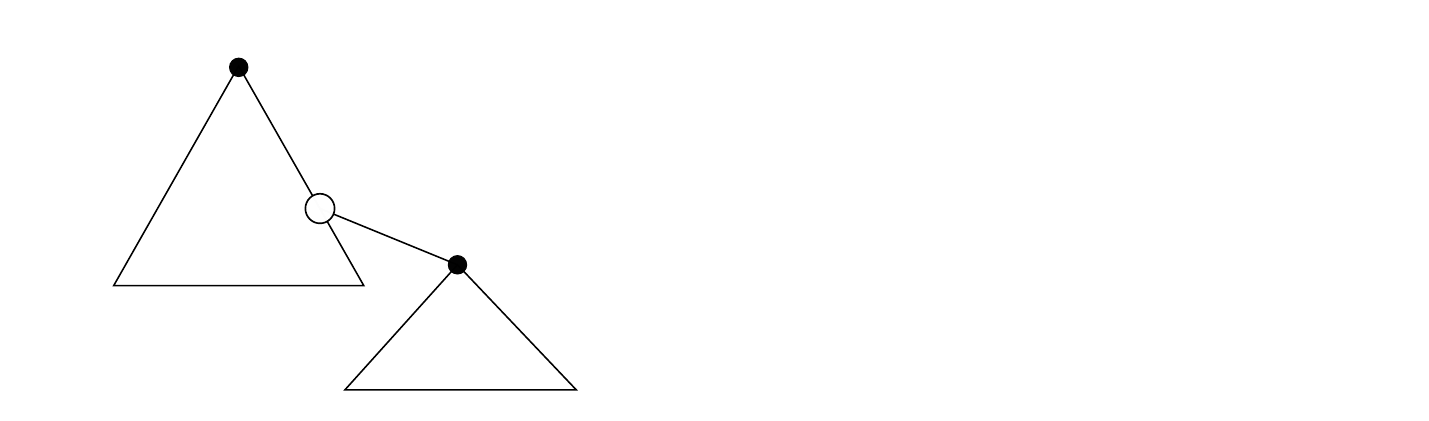} {inkscape -z -C --file=ststr12.svg --export-pdf=ststr12.pdf
--export-latex}
\begingroup
  \makeatletter
  \providecommand\color[2][]{
    \errmessage{(Inkscape) Color is used for the text in Inkscape, but the package 'color.sty' is not loaded}
    \renewcommand\color[2][]{}
  }
  \providecommand\transparent[1]{
    \errmessage{(Inkscape) Transparency is used (non-zero) for the text in Inkscape, but the package 'transparent.sty' is not loaded}
    \renewcommand\transparent[1]{}
  }
  \providecommand\rotatebox[2]{#2}
  \ifx\svgwidth\undefined
    \setlength{\unitlength}{411.42985667bp}
    \ifx\svgscale\undefined
      \relax
    \else
      \setlength{\unitlength}{\unitlength * \real{\svgscale}}
    \fi
  \else
    \setlength{\unitlength}{\svgwidth}
  \fi
  \global\let\svgwidth\undefined
  \global\let\svgscale\undefined
  \makeatother
  \begin{picture}(1,0.29407927)
    \put(0,0){\includegraphics[width=\unitlength,page=1]{ststr12.pdf}}
    \put(0.16321481,0.25731906){\color[rgb]{0,0,0}\makebox(0,0)[lb]{\smash{$u$}}}
    \put(0.22008956,0.16325688){\color[rgb]{0,0,0}\makebox(0,0)[lb]{\smash{$\white$}}}
    \put(0.15519408,0.13992358){\color[rgb]{0,0,0}\makebox(0,0)[lb]{\smash{$\RootedTreeOne$}}}
    \put(0.31633948,0.11950697){\color[rgb]{0,0,0}\makebox(0,0)[lb]{\smash{$\black_{\ttime}$}}}
    \put(0.04071467,0.15888162){\color[rgb]{0,0,0}\makebox(0,0)[lb]{\smash{$\RootedTreeTwo$:}}}
    \put(0.31196453,0.04950717){\color[rgb]{0,0,0}\makebox(0,0)[lb]{\smash{$\RootedTreeThree{}$}}}
    \put(0,0){\includegraphics[width=\unitlength,page=2]{ststr12.pdf}}
    \put(0.54598826,0.25758672){\color[rgb]{0,0,0}\makebox(0,0)[lb]{\smash{$u$}}}
    \put(0.60286291,0.16352454){\color[rgb]{0,0,0}\makebox(0,0)[lb]{\smash{$\white$}}}
    \put(0.45192608,0.23862868){\color[rgb]{0,0,0}\makebox(0,0)[lb]{\smash{$\RootedTreeOne$}}}
    \put(0.69911304,0.11977463){\color[rgb]{0,0,0}\makebox(0,0)[lb]{\smash{$\black_{\ttime}$}}}
    \put(0.39796688,0.15914927){\color[rgb]{0,0,0}\makebox(0,0)[lb]{\smash{$\RootedTreeTwo$:}}}
    \put(0.74286252,0.08039957){\color[rgb]{0,0,0}\makebox(0,0)[lb]{\smash{$\white'$}}}
    \put(0.67790181,0.13694416){\color[rgb]{0,0,0}\rotatebox{30}{\makebox(0,0)[lb]{\smash{$\treeDist{\RootedTreeTwo}{\black_{\ttime}} < \infty$}}}}
    \put(0.73664258,0.0971415){\color[rgb]{0,0,0}\rotatebox{30}{\makebox(0,0)[lb]{\smash{$\treeDist{\RootedTreeTwo}{\white'} < \infty$}}}}
  \end{picture}
\endgroup
  }
\caption{The case when $\black_{\ttime}$ has at least two neighbours which are alive.}
\label{fig:ststr12}
\end{figure}

Hence
$$
\treeDist{\RootedTreeTwo}{\white}=\max \set{ \treeDist{\RootedTreeOne}{\white}, \treeDist{\RootedTreeTwo}{\black_{\ttime}}+1} \  
\begin{cases}
= \infinity \text{ if } \treeDist{\RootedTreeOne}{\white}=\infinity,\\
< \infinity \text{ otherwise.}
\end{cases}
$$
To complete the proof it remains to show that $\black_{\ttime} \in \Alive_\ttime$. Let $\RootedTreeOne$ be a component of $\Forest_{\ttime}$ rooted at $\black_{\ttime}$. Note that $\dist{\ttime-1}{\white}=\treeDist{\RootedTreeOne}{\white}$ for $\white \in \children{\RootedTreeOne}{\black_{\ttime}}$. Since there exists $\white_1,\white_2 \in \children{\RootedTreeOne}{\black_{\ttime}}$ such that $\dist{\ttime-1}{\white_i}=\treeDist{\RootedTreeOne}{\white_i}  < \infinity$ for $i \in \{1,2 \}$, we have $\dist{\ttime}{\black_{\ttime}} < \infinity$ and $\secondDist{\ttime}{\black_{\ttime}}< \infinity$.
\end{proof}

\lemdispatching*

\begin{proof}
The illustration for this lemma is provided in Figure~\ref{fig:dead_parts}.
By monotonicity of $\dist{\ttime}{v}$ and $\secondDist{\ttime}{v}$ functions with respect to $\ttime$ (see Observation \ref{obs:monotonic}), the vertices dead in turn $\ttime-1$ remain dead in turn $\ttime$.
By Definitions~\ref{def:distdir}
and~\ref{def:deadBlack} vertex $\black_{\ttime}$ dies in turn $\ttime$.
By Observation~\ref{obs:aliveness} all other vertices on $\pathOne$ also die in turn $\ttime$.
What we need to prove is that the remaining vertices stay alive, i.e., do not die in turn $\ttime$.

For the sake of this proof we introduce notation for a rooted subtree. For a rooted tree $\RootedTreeOne$ and a vertex $u \in V(\RootedTreeOne)$ we let $\RootedSubTree{\RootedTreeOne}{u}$ denote the subtree of $\RootedTreeOne$ rooted at $u$ containing $u$ and all its descendants in $\RootedTreeOne$.
We first observe that $\lp{\ttime} \subseteq \Alive_{\ttime-1}$ due to Observation~\ref{obs:aliveness}. For $p \in \lp{\ttime}$ let $\WA{\ttime-1}{p}=\neighbours{\ttime-1}{p} \cap \Alive_{\ttime-1}$. Due to Definition~\ref{def:lifePortals}, $\size{\WA{\ttime-1}{p}} \geq 3$. We first prove that life portals remain alive, i.e., $\lp{\ttime}\subseteq \Alive_{\ttime}$. Let $\RootedTreeOne$ be the connected component of $\Forest_{\ttime-1}$ rooted in $p$ and let $\RootedTreeTwo$ be the connected component of $\Forest_{\ttime}$ rooted in $p$. Observe that $\WA{\ttime-1}{p} \subseteq \children{\RootedTreeOne}{p}$. There is only on vertex $\white \in \children{\RootedTreeOne}{p}$, for which $\RootedSubTree{\RootedTreeOne}{\white}\neq \RootedSubTree{\RootedTreeTwo}{\white}$. Hence, there exist two vertices $\white_1, \white_2 \in \WA{\ttime-1}{p}$, for whom 
$\RootedSubTree{\RootedTreeOne}{\white_i} = \RootedSubTree{\RootedTreeTwo}{\white_i}$ for $i \in \{1,2 \}$.
It then holds that for $i \in \{1,2 \}$: $$\infinity > \dist{\ttime-1}{\white_i} \geq \treeDist{\RootedTreeOne}{\white_i}=\treeDist{\RootedTreeTwo}{\white_i}.$$ Also,
$\children{\RootedTreeOne}{p} \setminus \{ \treeDir{\RootedTreeOne}{p} \}$
contains either $\white_1$ or $\white_2$, say it contains $\white_1$. Then
$$
\secondDist{\ttime}{p}=\min_{\white \in \children{\RootedTreeTwo}{p} \setminus \{ \treeDir{\RootedTreeOne}{p} \} } \treeDist{\RootedTreeTwo}{\white}+1 \leqslant \treeDist{\RootedTreeTwo}{\white_1}+1 < \infinity.
$$   
Now let $u \in \Black_{\ttime} \cup \White \setminus \lp{t}$ such that $u$ was alive in turn $\ttime-1$. It remains to prove that $u$ does not die in turn $\ttime$. If $u$ is not connected to $\black_{\ttime}$ in $\Forest_{\ttime}$, then $u$ obviously stays alive in turn $\ttime$.  Otherwise there is a unique path from $u$ to $\black_{\ttime}$ in $\Forest_{\ttime}$ and this path contains a life portal, let it be $p \in \lp{t}$. Let now $\RootedTreeOne$ be a component of $u$ in $\Forest_{\ttime-1}$ rooted in $u$ and let $\RootedTreeTwo$ be a component of $u$ in $\Forest_{\ttime}$ rooted in $u$.

Let $\Forest_{\ttime}- \set{\white, \black_{\ttime} }$ be the graph obtained by removing edge $\set{\white, \black_{\ttime}}$ from $\Forest_{\ttime}$. Let $\RootedTreeThree$ be the connected component of $\black_{\ttime}$ in $\Forest_{\ttime} - \set{\white, \black_{\ttime} }$, rooted at $\black_{\ttime}$. The only difference between $\RootedTreeOne$ and $\RootedTreeTwo$ is that $\black_{\ttime}$ becomes a child of some vertex $\white \in V(\RootedSubTree{\RootedTreeOne}{p})$ and $\RootedTreeThree$ is attached to $\RootedTreeOne$ via edge $(\white,\black_{\ttime})$. Also, in the proof of Observation~\ref{obs:monotonic} we proved that $\treeDist{\RootedTreeOne}{v} \leq \treeDist{\RootedTreeTwo}{v}$ for $v \in V(\RootedTreeOne)$.

Based on these observations, it suffices to show that 
$\treeDist{\RootedTreeOne}{p} < \infinity \Rightarrow \treeDist{\RootedTreeTwo}{p} < \infinity$.
For that it suffices to show that $\treeDist{\RootedTreeTwo}{p} < \infinity$. Now observe that $\WA{\ttime-1}{p} \cap \children{\RootedTreeTwo}{p} \geq 2$ (one of the vertices in $\WA{\ttime-1}{p}$ could be $p$'s parent). Let $\white_1,\white_2 \in \WA{\ttime-1}{p} \cap \children{\RootedTreeTwo}{p}$. Then there is $\white_i \in \{ \white_1,\white_2 \}$ such that $\RootedSubTree{\RootedTreeOne}{\white_i}= \RootedSubTree{\RootedTreeTwo}{\white_i}$. Since $\white_i \in \WA{\ttime-1}{p}$, it holds that $$ \infinity > \dist{\ttime-1}{\white_i} \geq \treeDist{\RootedTreeOne}{\white_i} = \treeDist{\RootedTreeTwo}{\white_i}.$$ This implies that $$\treeDist{\RootedTreeTwo}{p} = \min_{\white \in \children{\RootedTreeTwo}{p}} \treeDist{\RootedTreeTwo}{\white} + 1 \leq \treeDist{\RootedTreeTwo}{\white_i} + 1  < \infinity.$$  
\end{proof}

\clearpage

\section{The proof}\label{app:proof}

\obsdistPathDisp*

\begin{proof}
The statement of the observation holds by definition if $\disp{\ttime}=\black_{\ttime}$. Otherwise the only possible scenario is the one covered by Lemma~\ref{lem:dispatching}. Let $\white^p$ be the predecessor of $\disp{\ttime}$ on $\pathh{\ttime}{\black_{\ttime}}$. According to Lemma~\ref{lem:dispatching} it holds that $\white^p \in \Dead_{\ttime}$. By Lemma~\ref{lem:mmPathsAlive} path $\pathh{\ttime}{\disp{\ttime}}$ cannot visit $\white^p$. Therefore $\pathh{\ttime}{\disp{\ttime}}=\pathhs{\ttime}{\black_{\ttime}}$, and this implies the desired claim.
\end{proof}

This section is devoted to proving Lemma~\ref{lem:plusMinusOne}. The proof requires analyzing many trees rooted at the vertices of the considered mini-max paths. Such analysis becomes much simpler if we reformulate some of the definitions introduced in the paper. We start by reformulating the definition of the $\treeDist{}{}$ function.

\begin{definition}
\label{def:det-dist} For $\white\in\White$, $\black\in\Black_{\ttime}$ such that $\set{\white, \black } \in \InducedEdges{\White \cup \Black_{\ttime}}$ the determined mini-max distance is defined as
\begin{align*}
\detDist{ \ttime}{\white}{\black} = &\left\lbrace
\begin{array}
{ll} \min_{\white' \in \neighbours{\ttime}{\black} \setminus \set{\white}}
\detDist{\ttime}{\black}{\white'}+1 & \text{if} \ \neighbours{\ttime}{\black} \setminus
\set{\white}\neq \emptyset \\
\infinity & \text{otherwise,}
\end{array}
\right.\\
\detDist{ \ttime}{\black}{\white} = & \left\lbrace
\begin{array}
{ll} \max_{\black' \in \neighbours{\ttime}{\white} \setminus \set{\black}}
\detDist{\ttime}{\white}{\black'}+1 & \text{if} \ \neighbours{\ttime}{\white}
\setminus \set{\black}\neq \emptyset \\
0 & \text{otherwise.}
\end{array}
\right.
\end{align*}
 Vertex $\black'\in \neighbours{\ttime}{\white}\setminus \set{\black}$
which determines the maximum is denoted as $\detDir{\ttime}{\black}{\white}$.
If there are more such vertices we choose the first one in some predefined order on the vertices $\Black \cup \White$.
If $\neighbours{\ttime}{\white}\setminus \set{\black}$ is empty then
$\detDir{\ttime}{\black}{\white}$ is not defined.
 Vertex $\white'\in \neighbours{\ttime}{\black}\setminus \set{\white}$ which determines the
minimum is denoted as $\detDir{\ttime}{\white}{\black}$.
If there are more such vertices we choose the first one in some predefined order on the vertices $\Black \cup \White$.
If $\neighbours{\ttime}{\black}\setminus \set{\white}$ is empty then
$\detDir{\ttime}{\white}{\black}$ is not defined.
\end{definition}

It is easy to observe that if we take a connected component of a vertex $\firstVertex\in\White\cup\Black_{\ttime}$ and we indicate $\firstVertex$ as a root then we obtain a rooted tree $\RootedTreeOne$ and the following holds
$$
\treeDist{\RootedTreeOne}{\secondVertex}=\Cases{
\detDist{\ttime}{\parent{\RootedTreeOne}{\secondVertex}}{\secondVertex}
}{
if $\secondVertex\neq\firstVertex$,
}{
\dist{\ttime}{\secondVertex}
}{
otherwise
}
$$
for all $\secondVertex\in\Vertices\bra{\RootedTreeOne}$. As a consequence we obtain two observations that follow.
\begin{myobservation}
\label{obs:dist}
 For $\white\in\White$, $\black\in\Black_{\ttime}$ the mini-max distance equals
\begin{align*}
\dist{ \ttime}{\black} = &\left\lbrace
\begin{array}
{ll} \min_{\white' \in \neighbours{\ttime}{\black}}
\detDist{\ttime}{\black}{\white'}+1 & \text{if} \ \neighbours{\ttime}{\black} \neq \emptyset \\
\infinity & \text{otherwise.}
\end{array}
\right.\\
\dist{ \ttime}{\white} = & \left\lbrace
\begin{array}
{ll} \max_{\black' \in \neighbours{\ttime}{\white}}
\detDist{\ttime}{\white}{\black'}+1 & \text{if} \ \neighbours{\ttime}{\white}\neq \emptyset \\
0 & \text{otherwise.}
\end{array}
\right.
\end{align*}
Vertex $\white'\in \neighbours{\ttime}{\black}$ which determines the minimum equals $\dir{\ttime}{\black}$.
If there are more such vertices this is the first one in some predefined order on the vertices $\Black \cup \White$.
If $\neighbours{\ttime}{\black}$ is empty than $\dir{\ttime}{\black}$ is not defined.
Vertex $\black'\in \neighbours{\ttime}{\white}$ which determines the maximum equals $\dir{\ttime}{\white}$.
If there are more such vertices this is the first one in some predefined order on the vertices $\Black \cup \White$.
If $\neighbours{\ttime}{\white}$ is empty than $\dir{\ttime}{\white}$ is not defined.
\end{myobservation}

\begin{myobservation}
\label{obs:secdist}
 For $\white\in\White$, $\black\in\Black_{\ttime}$ the mini-max distance equals
\begin{align*}
\secondDist{ \ttime}{\black} = &\left\lbrace
\begin{array}
{ll} \min_{\white' \in \neighbours{\ttime}{\black}\setminus \set{\dir{\ttime}{\black}}}
\detDist{\ttime}{\black}{\white'}+1 & \text{if} \ \neighbours{\ttime}{\black}\setminus \set{\dir{\ttime}{\black}} \neq \emptyset \\
\infinity & \text{otherwise.}
\end{array}
\right.\\
\secondDist{\ttime}{\white} = & \left\lbrace
\begin{array}
{ll} \max_{\black' \in \neighbours{\ttime}{\white}\setminus \set{\dir{\ttime}{\white}}}
\detDist{\ttime}{\white}{\black'}+1 & \text{if} \ \neighbours{\ttime}{\white}\setminus \set{\dir{\ttime}{\white}}\neq \emptyset \\
0 & \text{otherwise.}
\end{array}
\right.
\end{align*}
 Vertex $\white'\in \neighbours{\ttime}{\black}\setminus \set{\dir{\ttime}{\black}}$ which determines the
minimum equals $\secondDir{\ttime}{\black}$.
If there are more such vertices this is the first one in some predefined order on the vertices $\Black \cup \White$.
If $\neighbours{\ttime}{\black}\setminus \set{\dir{\ttime}{\black}}$ is empty than
$\secondDir{\ttime}{\black}$ is not defined.
 Vertex $\black'\in \neighbours{\ttime}{\white}\setminus \set{\dir{\ttime}{\black}}$
which determines the maximum equals $\secondDir{\ttime}{\white}$.
If there are more such vertices this is the first one in some predefined order on the vertices $\Black \cup \White$.
If $\neighbours{\ttime}{\white}\setminus \set{\dir{\ttime}{\black}}$ is empty than
$\secondDir{\ttime}{\white}$ is not defined.
\end{myobservation}

The meaning of $\dir{\ttime}{\cdot}$ and $\secondDir{\ttime}{\cdot}$ functions is given by the following observation which follows from Definition \ref{def:det-dist} and Observations~\ref{obs:dist}~and~\ref{obs:secdist}.

\begin{myobservation}
\label{obs:distSecdist}
For each $\firstVertex\in\Black_{\ttime}\cup\White$ it holds that
\begin{align*}
\dist{\ttime}{\firstVertex}&=\detDist{\ttime}{\firstVertex}{\dir{\ttime}{\firstVertex}}+1,\\
\secondDist{\ttime}{\firstVertex}&=\detDist{\ttime}{\firstVertex}{\secondDir{\ttime}{\firstVertex}}+1.
\end{align*}
\end{myobservation}

Observation \ref{obs:distSecdist} provides a formula for $\dist{\ttime}{\firstVertex}$ and $\secondDist{\ttime}{\firstVertex}$ in terms of $\detDist{\ttime}{\firstVertex}{\cdot}$. If we want to determine $\detDist{\ttime}{\secondVertex}{\firstVertex}$ by $\dist{\ttime}{\firstVertex}$ and $\secondDist{\ttime}{\firstVertex}$ it is worth noting that $\dist{\ttime}{\firstVertex}$ almost always equals $\dist{\ttime}{\firstVertex}$. The only  exception is when $\dist{\ttime}{\firstVertex}$ is determined by $\secondVertex$ and equals $\secondDist{\ttime}{\firstVertex}$.

\begin{myobservation}
\label{obs:detdist}
For any edge $\set{\firstVertex,\secondVertex}\in\Edges_{\ttime}$ the following holds
$$
\detDist{\ttime}{\secondVertex}{\firstVertex}=
\Cases
{\dist{\ttime}{\firstVertex}}{if $\secondVertex\neq\dir{\ttime}{\firstVertex}$}
{\secondDist{\ttime}{\firstVertex}}{otherwise.}
$$
\end{myobservation}

Because $\dist{}{}$ is a $\min$ ($\max$) taken over a bigger set than $\secondDist{}{}$ we have the following.

\begin{myobservation}
\label{obs:distDetSecond}
For $\white\in\White$ and $\black\in\Black_{\ttime}$ such that $\set{\white,\black}\in\Edges_{\ttime}$ the following is holds
\begin{align*}
\dist{\ttime}{\black}\leqslant&\detDist{\ttime}{\white}{\black}\leqslant\secondDist{\ttime}{\black},\\
\secondDist{\ttime}{\white}\leqslant&\detDist{\ttime}{\black}{\white}\leqslant\dist{\ttime}{\white}.
\end{align*}
\end{myobservation}

The function $\detDist{ \ttime}{\firstVertex}{\secondVertex}$, given an entry edge
$\dedge{\firstVertex}{\secondVertex}$ to a vertex $\secondVertex$, returns one plus the distance
of some neighbor $\detDir{\ttime}{\firstVertex}{\secondVertex}$ of $\secondVertex$ other than
$\firstVertex$.
This defines a determined mini-max path in a natural recursive way.
An undetermined mini-max path can be defined using the notion of a determined mini-max path.

\begin{definition}
\label{def:mmpaths}
For any $\ttime=0, \ldots,n$ and $\set{\firstVertex,\secondVertex}\in\Edges_{\ttime}$ we define
$$
\detPath{\ttime}{\firstVertex}{\secondVertex}= \left \lbrace
\begin{array}
{ll} \firstVertex \cdot
\detPath{\ttime}{\secondVertex}{\detDir{\ttime}{\firstVertex}{\secondVertex}} & \text{if}
\ \detDir{\ttime}{\firstVertex}{\secondVertex} \ \text{exists,} \\
\firstVertex \cdot \secondVertex& \text{otherwise.}
\end{array}
\right.
$$
\end{definition}
We observe the following.
\begin{myobservation}
Let $\ttime=0, \ldots,n$ and $\firstVertex\in\Black_{\ttime}\cup\White$. Then
\begin{align*}
\pathh{\ttime}{\firstVertex}&= \left \lbrace
\begin{array}
{ll} \detPath{\ttime}{\firstVertex}{\dir{\ttime}{\firstVertex}} & \text{if}
\ \dir{\ttime}{\firstVertex} \ \text{exists,} \\
\firstVertex& \text{otherwise.}
\end{array}
\right.\\
\secondPath{\ttime}{\firstVertex}&= \left \lbrace
\begin{array}
{ll} \detPath{\ttime}{\firstVertex}{\secondDir{\ttime}{\firstVertex}} & \text{if}
\ \secondDir{\ttime}{\firstVertex} \ \text{exists,} \\
\firstVertex& \text{otherwise.}
\end{array}
\right.
\end{align*}
\end{myobservation}

The next observation shows a relation between mini-max functions and the lengths of the corresponding mini-max paths.

\begin{myobservation}
\label{obs:lenpaths}
For any $\ttime=0, \ldots,n$ and any edge $\set{\firstVertex,\secondVertex}\in\Edges_{\ttime}$ and $\thirdVertex\in\White\cup\Black_{\ttime}$ we have that
\begin{align*}
\length{\detPath{\ttime}{\firstVertex}{\secondVertex}}&=\detDist{\ttime}{\firstVertex}{\secondVertex} &\text{if } \detDist{\ttime}{\firstVertex}{\secondVertex} < \infinity,\\
\length{\pathh{\ttime}{\thirdVertex}}&=\dist{\ttime}{\thirdVertex} &\text{if } \dist{\ttime}{\thirdVertex} < \infinity,\\
\length{\secondPath{\ttime}{\thirdVertex}}&=\secondDist{\ttime}{\thirdVertex} &\text{if } \secondDist{\ttime}{\thirdVertex} < \infinity.
\end{align*}
\end{myobservation}

We are now ready to prove the following interesting property that says that the level of the $\dir{}{}$ and $\secondDir{}{}$ can not change by more than $\pm1$. 

\lemplusMinusOne*

\begin{proof}
We present the proof for $\firstVertex=\white\in\White$ and then $\secondVertex=\black\in\Black_{\ttime}$. The case when $\firstVertex\in\Black_{\ttime}$ is symmetric.
We split this case to the following four subcases.

\begin{description}
\item [First subcase ]
We assume that $\dir{\ttime}{\white}=\black$ and $\dir{\ttime}{\black}=\white$.
Then
\begin{multline*}
\level{\ttime}{\black}=^{\text{Def.\ref{def:level}.}}
\dist{\ttime}{\black}=^{\text{Obs.\ref{obs:distSecdist}.}}
\detDist{\ttime}{\black}{\white}+1=^{\text{Obs.\ref{obs:detdist}.}}
\\=
\secondDist{\ttime}{\white}+1=^{\text{Def.\ref{def:level}.}}
\level{\ttime}{\white}+1.
\end{multline*}

\item [Second subcase ]
We assume that $\dir{\ttime}{\white}=\black$ and $\dir{\ttime}{\black}\neq\white$.
Then
\begin{multline*}
\level{\ttime}{\white}=^{\text{Def.\ref{def:level}.}}
\secondDist{\ttime}{\white}\leqslant^{\text{Obs.\ref{obs:distDetSecond}.}}
\dist{\ttime}{\white}=^{\text{Obs.\ref{obs:distSecdist}.}}
\\=
\detDist{\ttime}{\white}{\black}+1=^{\text{Obs.\ref{obs:detdist}.}}
\dist{\ttime}{\black}+1=^{\text{Def.\ref{def:level}.}}
\level{\ttime}{\black}+1.
\end{multline*}
On the other side we have that
\begin{multline*}
\level{\ttime}{\black}=^{\text{Def.\ref{def:level}.}}
\dist{\ttime}{\black}=^{\text{Obs.\ref{obs:dist}}}
\min_{\white'\in\neighbours{\ttime}{\black}}\detDist{\ttime}{\black}{\white'}+1\leqslant
\\=
\detDist{\ttime}{\black}{\white}+1=^{\text{Obs.\ref{obs:detdist}.}}
\secondDist{\ttime}{\white}+1=^{\text{Def.\ref{def:level}.}}
\level{\ttime}{\white}+1.
\end{multline*}

\item[Third subcase ]
We assume that $\secondDir{\ttime}{\white}=\black$ and $\dir{\ttime}{\black}=\white$.
Then
\begin{multline*}
\dist{\ttime}{\black}+2\leqslant^{\text{Obs.\ref{obs:distDetSecond}.}}
\detDist{\ttime}{\white}{\black}+2=^{\text{Obs.\ref{obs:distSecdist}.}}
\secondDist{\ttime}{\white}+1\leqslant^{\text{Obs.\ref{obs:distDetSecond}.}}
\\\leqslant
\detDist{\ttime}{\black}{\white}+1=^{\text{Obs.\ref{obs:distSecdist}.}}
\dist{\ttime}{\black}
\end{multline*}
but it is not possible.

\item[Fourth subcase ]
The last case describes situation when $\secondDir{\ttime}{\white}=\black$ and $\dir{\ttime}{\black}\neq\white$. Then
\begin{multline*}
\level{\ttime}{\white}=^{\text{Def.\ref{def:level}.}}
\secondDist{\ttime}{\white}=^{\text{Obs.\ref{obs:distSecdist}.}}
\detDist{\ttime}{\white}{\black}+1=^{\text{Obs.\ref{obs:detdist}.}}
\\=
\dist{\ttime}{\black}+1=^{\text{Def.\ref{def:level}.}}
\level{\ttime}{\black}+1.
\end{multline*}
\end{description}
\end{proof}

\end{document}